\documentclass[12pt,centertags,reqno]{amsart}

\usepackage[foot]{amsaddr}
\usepackage{latexsym}
\usepackage[english]{babel}
\usepackage[T1]{fontenc}
\usepackage[numbers]{natbib}
\usepackage{amssymb}
\usepackage{fancyhdr}
\usepackage{url}
\usepackage{hyperref}
\usepackage{verbatim}
\usepackage{leftidx}
\usepackage{color,graphicx}
\usepackage{pdfpages}

\usepackage{mathrsfs, mathtools}
\usepackage{stmaryrd}
\usepackage{nicefrac}
\usepackage{marvosym}
\mathtoolsset{showonlyrefs}

\numberwithin{equation}{section} \makeatletter
\renewcommand{\subsection}{\@startsection
	{subsection}{2}{0mm}{\baselineskip}{-0.25cm}
	{\normalfont\normalsize\bf}} \makeatother

\newtheorem{theorem}{Theorem}[section]
\newtheorem{lemma}[theorem]{Lemma}

\newtheorem{proposition}[theorem]{Proposition}
\newtheorem{ass}[theorem]{Assumption}
\newtheorem{definition}[theorem]{Definition}
\theoremstyle{remark}
\newtheorem{remark}[theorem]{Remark}
\newtheorem{example}[theorem]{Example}

\def \F {\mathcal F}
\def \G {\mathcal G}
\def \H {\mathcal H}
\def \L {\mathcal L}

\def \P {\mathbf P}
\def \Q {\mathbf Q}
\def \I {{\mathbf 1}}

\def \R {\mathbb R}
\def \bF {\mathbb F}
\def \bG {\mathbb G}
\def \bH {\mathbb H}

\def \bN {\mathbb N}

\newcommand{\ud}{\mathrm d}

\newcommand{\esp}[1]{\mathbb{E}^\Q \left[#1\right]}

\newcommand{\CL}{\operatorname{CL}}
\newcommand{\CVA}{\operatorname{CVA}}
\newcommand{\CDS}{\operatorname{CDS}}

\begin{document}
	\author[C.~Ceci]{Claudia  Ceci}\address{Claudia  Ceci, Department of Economics, University ``G. D'Annunzio'' of Chieti-Pescara, Viale Pindaro, 42, 65127 Pescara, Italy.}\email{c.ceci@unich.it}
	
	\author[K.~Colaneri]{Katia Colaneri}\address{Katia Colaneri, Department of Economics and Finance, University of Rome Tor Vergata, Via Columbia 2, 00133 Rome, Italy.}\email{katia.colaneri@uniroma2.it}
	
	\author[R.~Frey]{R\"udiger Frey}\address{R\"udiger Frey, Institute for Statistics and Mathematics, Vienna University of Economics and Business, Welthandelsplatz, 1, 1020 Vienna, Austria }\email{rfrey@wu.ac.at}
	
	\author[V.~K\"ock]{Verena K\"ock}\address{Verena K\"ock, Institute for Statistics and Mathematics, Vienna University of Economics and Business, Welthandelsplatz, 1, 1020 Vienna, Austria }\email{verena.koeck@wu.ac.at}

	\title[Hedging reinsurance counterparty risk]{{Value adjustments and dynamic hedging of reinsurance counterparty risk}}

	\date{}
	
	\begin{abstract}
		Reinsurance counterparty credit risk (RCCR) is the risk of a loss arising from the fact that  a reinsurance company is unable to fulfill her contractual obligations towards the ceding insurer. RCCR  is an important risk category for insurance companies which, so far, has  been addressed mostly via qualitative approaches. %
In this paper we therefore study value adjustments and dynamic hedging for  RCCR.     We propose a novel model  that accounts for contagion effects between the default of the reinsurer and the price of the reinsurance contract.  We characterize the value adjustment in a reinsurance contract via a partial integro-differential equation (PIDE) and derive the hedging strategies using a quadratic method. The paper closes with a simulation study which shows that dynamic hedging strategies have the potential to significantly reduce RCCR.
		
	\end{abstract}

	\maketitle
	

\section{Introduction}
	
	General insurers frequently cede parts of their insurance risk to reinsurance companies  in order to protect themselves from intolerably large losses in their insurance portfolio. This  gives rise to a new type of risk, so-called \emph{reinsurance counterparty credit risk} or RCCR.
This is  the risk of a loss for the ceding  company caused by the fact that the  reinsurer  fails to honor her obligations from a  reinsurance contract, for instance because the reinsurer  defaults prior to maturity of the contract. Given the increased visibility of default risk in the reinsurance industry in the aftermath of the financial crisis, RCCR has become a highly  relevant risk category, mainly because reinsurance recoveries represent large assets on insurance companies balance sheets. Its importance is also underlined in Solvency II regulatory directives. Nonetheless, the  techniques for managing RCCR used in practice are mostly of a qualitative nature. Typically, ceding companies have  minimum requirements on the credit quality of approved  reinsurance companies, they set limits for the exposure to individual counterparties, and  they sometimes require reinsurers to post some  collateral; see for instance~\citet{bib:bodoff-13}.
The existing quantitative approaches for the management of RCCR are based on simple one-period models.  This is in stark contrast to the banking world  where  sophisticated stochastic  models are used in counterparty risk management to determine value adjustments for derivative transactions (so-called XVAs)   and  to find dynamic hedging and collateralization strategies, see for instance \citet{bib:gregory-12} or \citet{bib:brigo-et-al-13} for an overview.
	
In this paper we  explore the potential of  dynamic risk management techniques for reinsurance counterparty risk. Our objective is twofold: we discuss the computation of value adjustments to account for reinsurance default when pricing a contract, and we analyse dynamic hedging strategies in view of reducing the risk exposure. In fact, counterparty risk towards a major reinsurance company is a low-frequency, high-severity event so that bearing this risk is not attractive for the ceding company.

We consider a setting that is tailored to the analysis of RCCR. We model the aggregate claim amount process  $L$ underlying the reinsurance contract under consideration by a doubly stochastic compound Poisson process.   To capture  the effect that ``reinsurance companies are most likely to default in times of market stress, that is exactly when cedants are most reliant upon their reinsurance covers'' (\citet{bib:flower-07}), we introduce  several sources of dependence between  the aggregate claim amount $L$ and the default process $H^R$  of the reinsurance company.  There is positive correlation between the claim arrival intensity $\lambda^L$ and the default intensity $\lambda^R$ of the reinsurer; moreover,  $\lambda^L$ exhibits a contagious  jump at the default time $\tau_R$ of the reinsurer. 
In line with the concept of market consistent valuation we define the credit value adjustment (CVA) for a  reinsurance contract as the expected discounted value of the replacement cost for the contract  incurred by the insurer at the default time $\tau_R$. Using mathematical results from the companion paper~\citet{bib:colaneri-frey-19}, we characterize the CVA as classical  solution of an  partial integro-differential equation. Next we address the hedging of RCCR by dynamic trading in a credit default swap (CDS) on the reinsurance company. Here we resort to a quadratic hedging approach~(\citet{schweizer2001guided}), since perfect replication is  not possible. To determine  the hedging strategy we make use of an orthogonal decomposition of the CVA into a hedgeable and a non-hedgeable part, based on the  Galtchouk-Kunita-Watanabe decomposition of the associated discounted gains process.	
	The paper closes with a simulation study. We analyse the impact of model parameters on the size of the CVA and we compare  the performance of various hedging strategies.  Our numerical experiments show that dynamic CDS hedging strategies   significantly reduce reinsurance counterparty risk, both compared to a static hedging strategy (a strategy where the CDS position is not adjusted) and to the case where the insurance company does not hedge at all. More generally, the results suggest that dynamic risk-mitigation techniques can be very useful tools in the management of reinsurance counterparty risk.

We continue with a discussion of the existing literature. The quantitative literature on RCCR is relatively scarce. Interesting contributions from practitioners include \citet{bib:shaw-07} or
\citet{bib:flower-07} who propose a static model to assess the distribution of the RCCR loss, which  can be used for reserving and economic capital purposes. They employ  corporate bonds and CDSs to estimate reinsurance default rates and model correlation between defaults by reinsurers' asset return correlations. Another example is offered by \citet{bib:kravych-shevchenko-11} who study the problem of optimising the weight of  different reinsurance companies in a given reinsurance program in order to minimize the expected loss due to RCCR. Also the solvency capital requirement for RCCR under the Solvency II standard formula is computed from a simple one-period credit risk model, see for instance \citet{bib:ceiops-09b}. On the academic side, \citet{bernard2012impact}  and \citet{cai2014optimal} study how the possibility of a default of the reinsurer affects the form of optimal reinsurance contracts. An excellent overview of counterparty risk management in banking is given in  \citet{bib:gregory-12} or \citet{bib:brigo-et-al-13}.   Other recent contributions are, for instance \citet{crepey2015bilateral, crepey2015bilateralCVA, bo2019locally}.
Quadratic hedging criteria such as mean variance hedging and risk minimization have been applied in the insurance framework mainly for hedging life insurance contracts (e.g. unit linked contracts). Some recent references are, for instance \citet{bib:moller-01, dahl2006valuation, vandaele2008locally, biagini2017risk, ceci2017unit}
	
The rest of the paper is organized as follows. In Section \ref{sec:market model} we introduce and develop the modelling framework and discuss the different forms of interaction between the insurance and the reinsurance companies that are captured by our setting. A rigorous construction of the model dynamics is provided in Section \ref{subsec:construct}. In Section \ref{sec:CVA} we discuss the price of the reinsurance contract and the value adjustment to account for the reinsurer default. The hedging problem is studied in Section \ref{sec:hedging}, and Section \ref{sec:numerics} contains the results from the  numerical analysis. Some longer computations are relegated to the Appendix.

	\section{The  model}\label{sec:market model}
	\subsection{The Setup}\label{subsec:setup}
	
	We work on  a measurable space $(\Omega,\G)$ with a  complete and right continuous filtration $\bG=(\G_t)_{t \ge 0}$. We assume that on this space there are two equivalent probability measures: the physical measure $\P$ and a risk neutral measure $\Q$ which is used for the valuation of  financial and actuarial contracts.  Using a risk-neutral measure for  pricing purposes  is  in line with the principle of  {\em market consistency valuation}, which is frequently used in the insurance framework and which represents one of the core elements of the Solvency II regulatory regime.
	
	We consider a setup with two companies: an insurance company, labelled $I$, and a reinsurer $R$, who enter into a reinsurance contract with a given maturity $T$ (typically one year).   To model the losses in the  insurance portfolio underlying this contract we consider a sequence $\{T_n\}_{n \in \mathbb N}$ of claim arrival times and a sequence $\{Z_n\}_ {n\in \mathbb N}$ of claim sizes. Precisely,  the $T_n$ are $\bG$-stopping times such that $T_n < T_{n+1}$  a.s. and  $Z_n$ are a.s. strictly positive $\mathcal{G}_{T_n}$-measurable random variables. We define the counting process $N=(N_t)_{t \geq 0}$ by $N_t = \sum_{n=1}^{\infty} \I_{\{T_n \le t\}}$, for every $t \geq 0$. Then the process $L=(L_t)_{t \geq 0}$ given by
	\begin{align}\label{c1}
	L_t=\sum_{n=1}^{N_t}Z_n\,, \quad t \geq 0,
	\end{align}
	describes  the aggregate claim amount underlying the reinsurance contract. It will be convenient to work with the  integer-valued random measure $m^L$ on $\R^+ \times \R^+$ associated with the marked point process $L$, that is
	\begin{equation}\label{m} m^L(\ud t, \ud z) = \sum_{n \geq 1} \delta_{\{T_n, Z_n\}} (\ud t, \ud z) \I_{\{ T_n < \infty\}},  \end{equation}
	where $\delta_{\{t,z\}}$ is the Dirac measure at point $(t,z) \in \R^+ \times \R^+$ .
	This allows for the following equivalent expression of $L$
	$$
	L_t=\int_0^t\int_{\R^+} z \ m^L(\ud s, \ud z), \quad t \ge 0.
	$$
In our setting the reinsurance company may default and we denote by $\tau_R$ the $\bG$-stopping time representing its default time; the  default indicator process $H^R=(H^R_t)_{t \geq 0}$ is given by
\begin{align}\label{eq:default_proc}
	H_t^R=\I_{\{\tau_R\le t\}}, \quad t \ge 0.
\end{align}
	If $\tau_R\le T$, the reinsurer  will not be able to  fulfill his obligations which  creates reinsurance counterparty credit risk (RCCR).

Next  we specify the  model for the loss process $L$ and the default indicator $H^R$.   In our analysis  we are mostly concerned with valuation issues so we work  under  the risk-neutral measure $\Q$; to simplify the exposition we therefore introduce directly the $\Q$ dynamics of $L$ and $H^R$.
	We assume that the point  process $N$ modeling the claim arrivals  has the $(\bG, \Q)$-intensity $\lambda^L$ for a   nonnegative $\bG$-adapted  c\'{a}dl\'{a}g process $\lambda^L=(\lambda^L_t)_{t \ge 0}$ (called in the sequel {\em loss intensity}), that is $(N_t - \int_0^t \lambda_{s-}^L \ud s)_{t \ge 0}$ is a $(\bG, \Q)$-martingale. We assume that claim sizes 
are independent random variables with identical distribution $\nu(\ud z)$, and also independent of $N$.
Therefore the  $(\bG, \Q)$-predictable compensator of the measure $m^L (\ud t, \ud z)$ is given by $\lambda^L_{t-}\nu(\ud z)\ud t$ \footnote{By definition of $(\bG, \Q)$-predictable compensator, for every  nonnegative, $\bG$-predictable random function  $(\Gamma(t, z))_{t \geq 0}$ with  $\esp{\int_0^t \int_{\mathbb R^+}|\Gamma(s,z)| \lambda^L_{s-} \nu(\ud z)\ud s} < \infty$, for every $t \geq 0$,  the process
    $$ \int_0^t \int_{\mathbb R^+} \Gamma(s,z)   ( m^L(ds,dz) - \lambda^L_{s-} \nu(\ud z)\ud s), \quad  t  \ge 0,
	$$
is a $(\bG, \Q)$-martingale.}.
We assume that the default indicator process $H^R$ admits  a stochastic intensity $\lambda^R=(\lambda^R_t)_{t \geq 0}$ (in the sequel called the \emph{default intensity} of $R$), which is  a  nonnegative $\bG$-adapted c\'{a}dl\'{a}g process  
	such that the process
	\begin{equation} \label{mg2}
	M^R_t := H^R_t - \int_0^{t\wedge \tau_R}  \lambda^R_{s-} \ud s, \quad t \ge 0,
	\end{equation}
	is a $(\bG, \Q)$-martingale. Finally we describe  the dynamics of the default and  the claim arrival intensity. We  assume that there is a standard two-dimensional $(\bG, \Q)$-Brownian motion $W=(W^1_t, W^2_t)_{t \geq 0}$ and that the processes $\lambda^L$ and $\lambda^R$ are of the form $\lambda_t^L=\lambda^L(X_t)$, $\lambda_t^R=\lambda^R(Y_t)$,  $t\ge 0$,   where $X=(X_t)_{t \geq 0}$ and $Y=(Y_t)_{t \geq 0}$ are intensity-factor processes that satisfy the following system of SDEs
	\begin{align}\label{eq1}
	\ud X_t&=  \gamma^X(X_{t-})  \  \ud H^R_t + b^X( X_t) \ud t + \sigma^X(X_t)\ud W^1_t, \quad X_0 =x_0\in \R\\
	\ud Y_t&=b^Y(Y_t) \ud t + \sigma^Y(Y_t)(\rho \ud W_t^1 + \sqrt{1-\rho^2}\ud W^2_t),\quad  Y_0=y_0\in \R,
	\end{align}
	for some $\rho \in [0,1]$ and measurable functions $b^X, b^Y:\R\to \R$, $\sigma^X,\sigma^Y:\R\to\R^+$.
	We assume that the functions $\lambda^L:\R \to \R^+$ and  $\lambda^R:\R \to \R^+$ and $\gamma^X:\R \to \R^+$ are continuous and increasing. A detailed construction of this model is given in  Section \ref{subsec:construct}. Modelling $\lambda^L$ and $\lambda^R$ as functions of the intensity factors $X$ and $Y$  is mathematically convenient, see Remark~\ref{rem:intensities} below.
	
We assume that the indemnity payment of the reinsurance contract is of  the form $\phi(L_T)$ for some bounded, increasing and Lipschitz continuous function $\phi$. This covers typical forms of reinsurance (see, e.g. \citet{bib:albrecher-et-al-17}). For examples, for a \emph{stop loss reinsurance} contract with {priority} or {lower attachment point} $\underline{K}$ and {upper limit} $\overline{K}$  one takes $\phi(l) = \min\big \{\overline K, [l-\underline{K}]^+\big \}$ (with  $[x]^+ = \max\{x,0\}$). Another example is offered by the \emph{excess-of-loss} (XL) contract with retention level $M$ and  upper limit $\overline{K}$. The payoff of this contract is given by
	$$ \min \Big \{ \overline{K}, \sum_{n=1}^{N_T} [Z_i - M]^+ \Big \}. $$
	This can be written in the form $\phi(L^{\text{XL}}_T)$  if we set  $	L^{\text{XL}}_t = \vspace{-0.2cm}\sum_{\{T_n \le t, Z_n >M \}}[Z_n - M]$ and $\phi(l) = \min \{\overline K, l\}$.
	
We denote by $r\ge 0$ the risk-free interest rate  which is taken constant for simplicity. In line with market consistent valuation we  define the \emph{market value} of the reinsurance contract  by
	\begin{equation}\label{eq:market-value}
	V_t^\phi :=\esp{e^{-r(T-t)}\phi(L_T)|\G_t}, \quad  0 \leq t \leq T\,.
	\end{equation}
The quantity $V_t^\phi$ is the theoretical (counterparty-risk free) value of the reinsurance contract at time $t$.  Due to the fact that the reinsurer R may default, the  transaction price (the price at which I and R are actually entering into the contract)  needs to be adjusted. This is done via the credit value adjustment introduced in Section~\ref{sec:CVA}.

	\begin{remark}[Market consistent valuation] \label{rem:market-consistent-valuation}
		The use of market consistent valuation does not mean that $R$ and $I$ are risk-neutral. In our context a risk premium  can be built into the model by choosing  the claim arrival intensity to be larger under $\Q$ than under $\P$;   by distorting  the claim size distribution, making large claims more likely under $\Q$ than under $\P$; and finally by assuming that the risk-neutral default intensity  is larger than the $\P$-default intensity.
	\end{remark}

	\begin{remark}[Dependence and pricing contagion]\label{rem:wwrisk} Our setting accounts for various forms of dependence between  the default intensity  of the reinsurer and the dynamics of the aggregate claim amount.  First,  there is  correlation between  Brownian motions, modelled by the parameter $\rho$. In practice one would take $\rho >0$, so that in a scenario where  the insurance company experiences many losses (high claim arrival intensity $\lambda^L$), the economic outlook  for the reinsurance company gets less favourable (high default intensity $\lambda^R$). This reflects  the fact that ``often there are strong correlations between reinsurance default and the loss experience of the ceded portfolio'' (\citet{bib:flower-07}). Second, there  is {\em pricing contagion}.  For $\gamma^X > 0$,  the risk-neutral claim arrival intensity $\lambda^L$ jumps upward at $\tau_R$,  which translates into an upward jump of the market value $V_t^\phi$  of the reinsurance contract at $t=\tau_R$, since claim arrival frequency becomes larger. This models  the fact that the  default of $R$ will reduce the supply for reinsurance, so that the market price of reinsurance goes up. We emphasize that  this is a pure pricing  phenomenon:   we do not claim that the default of $R$ has an effect on the claim arrival intensity under the real-world measure.
		Note that each of these two forms of dependence between $L$ and $\tau_R$ imply that $\esp{V^{\phi}_t|\tau_R=t}>\esp{V^{\phi}_t}$. In the financial literature on counterparty risk this inequality is known as {\em wrong-way risk}.
	\end{remark}

We now introduce a set of assumptions that give sufficient conditions for existence and uniqueness for the solutions of certain  partial integro-differential equations that arise in the computation of the value adjustment and of the hedging strategy.
Define the instantaneous covariance matrix of $(X,Y)$ as
$$\Sigma(x,y):=\left(\begin{array}{cc}  (\sigma^X(x))^2 & \rho \sigma^X(x) \sigma^Y(y) \\ \rho \sigma^X(x)\sigma^Y(y) & (\sigma^Y(y))^2  \end{array}\right) \ \mbox{ for every } \ (x,y)\in \R^2.$$
	\begin{ass}\label{ass:regularity}
		\begin{itemize}
			\item[(A1)] The functions $b^X, b^Y, \sigma^X$ and $\sigma^Y$ are Lipschitz;
			\item[(A2)] There exists $\beta>0$ such that for every $ w \in \R^2$ we have
			\begin{align}w^\top \Sigma(x,y)w\geq\beta \| w\|^2.\end{align}
			\item[(A3)] The functions $\lambda^L, \lambda^R$ are Lipschitz continuous and bounded.
			\item[(A4)] The claim-size distribution $\nu$ has finite second moment.
		\end{itemize}
	\end{ass}

\begin{remark}~\label{rem:intensities} Instead of modelling the loss and default intensities as functions of the stochastic factors $X$ and $Y$, one might model directly the processes $\lambda^L$ and $\lambda^R$. However, in this case it would be problematic to verify Assumption (A2), since intensities are forced to be nonnegative and therefore ellipticity of $\Sigma(x,y)$  does not hold.  	
\end{remark}	
	
	\subsection{Model construction}\label{subsec:construct}
	The goal of this section is to provide a step-by-step construction of the model introduced in Section \ref{sec:market model}. Moreover, we establish certain mathematical properties that are needed for the characterization of the credit value adjustment.
 We start by fixing a filtered probability space $(\Omega, \G, \Q)$. Let $W=(W_t)_{t \geq 0}$ be a two-dimensional Brownian motion with components $(W^1_t, W^2_t)_{t \geq 0}$, let  $\eta=(\eta_t)_{t \geq 0}$ be a standard Poisson process independent of $W$, and $\{Z_n\}_{n \in \bN}$ be a sequence of independent random variables with identical distribution $\nu(\ud z)$, and that are also independent of  $W$ and $\eta$.
	Define the process $M=(M_t)_{t \geq 0}$ with $M_t=\sum_{n=1}^{\eta_t}Z_n$. 
	This is a compound Poisson process  with intensity equal to one and jump size distribution $\nu(\ud z)$.
	Let the process $Y$ be the unique solution of the SDE
	\begin{align}
	\ud Y_t&=b^Y(Y_t) \ud t + \sigma^Y(Y_t)(\rho \ud W^1_t + \sqrt{1-\rho^2}\ud W^2_t),\quad  Y_0=y_0\in \R.
	\end{align}
	We assume that there exists a $\G$-measurable random variable  $\vartheta$ with unit exponential law, independent of $W$ and $M$ and we define  $\tau_R$ as
	\begin{align}
	\tau_R:=\inf\Big\{t\geq 0: \ \int_0^t \lambda^R(Y_s) \ud s\geq \vartheta\Big\}.
	\end{align}
	By construction the random time $\tau_R$ is doubly stochastic with respect to the filtration  $\bF^W\vee \bF^M$ with hazard rate $(\lambda^R(Y_t))_{t \geq 0}$,  that is we have for every $t >0$
	\begin{equation}\label{eq:default_probability}
	\Q( \tau_R > t \mid \F^W_\infty \vee\F^M_\infty) = \Q\Big( \int_0^t \lambda^R(Y_s) \ud s\leq \vartheta  \mid \F^W_\infty \vee\F^M_\infty \Big)= e^{-\int_0^t  \lambda^R( Y_s) \ud s};
	\end{equation}
see, e.g.  \citet[Section 8.2.1]{bielecki2004credit} or \citet[Section~10.5]{bib:mcneil-frey-embrechts-15} for details. We define   $H_t^R=\I_{\{\tau_R\leq t\}}$, $t \ge 0$, and  
we introduce the process $X$ as the unique solution to the SDE
	\begin{align}
	\ud X_t&=  \gamma^X(X_{t-})  \  \ud H_t^R + b^X( X_t) \ud t + \sigma^X(X_t)\ud W^1_t, \quad X_0 =x_0\in \R.
	\end{align}
To construct the aggregate claims process we use a time change argument.  Define the process $\theta=(\theta_t)_{t \geq 0}$ by  $\theta_t:=\int_0^t\lambda^L(X_{s-})\ud s$ for every $t \geq 0$ and let $N_t:=\eta_{\theta_t},$  $t \geq 0$. It is easy to see that $N=(N_t)_{t \geq 0}$ is a doubly stochastic point process with intensity $(\lambda^L(X_{t}))_{t \geq 0}$  (see, e.g. \citet{grandell2012aspects}) and that the loss process is given by  $L_t=M_{\theta_t}=\sum_{n=1}^{N_t}Z_n$.
	Finally  we define the filtration $\bG=(\G_t)_{t \geq 0}$ by
	\begin{align}\label{eq:filtrG}\G_t= \F^W_t\vee\F^L_t\vee\H_t, \quad t \geq 0,\end{align}completed with $\Q$-null sets, where $\bH=(\H_t)_{t \geq 0}$ is the natural filtration of the process $H^R$.
Notice that the random variables $Z_n$ are $\G_{T_n}$-measurable, with  $\{T_n\}_{n \in \bN}$ being the sequence of jump times of $N$. Moreover, $\tau_R$ is a stopping time with respect to the filtration $\bG$. We also have that $M^R$ in equation \eqref{mg2} is  $(\bG,\Q)$-martingale. This is a consequence of the fact that $M^R$ is a martingale with respect to the filtration $\bF^W\vee \bH$ and, due to independence between $M$ and $W$, this is also a martingale with respect to filtration $\bF^W\vee \bH \vee \F^M_{\infty}$. Now, since $\F^W_t\vee\F^{L}_t\vee \H_t\subset \F^W_t\vee\F^M_{\infty}\vee \H_t$ for every $t \geq 0$, then we have that the martingale property for $M^R$ holds for the filtration $\bF^W\vee  \bF^L \vee \bH$.

	
\paragraph{\bf The contagion-free market.}
In the remaining part of this section we introduce the {\em contagion-free} setting which will be used in the computations of the credit value adjustment and of the hedging strategies.
	Let $\widetilde X=(\widetilde X_t)_{t \geq 0}$ be the unique solution to the SDE
	\begin{align}
	\ud \widetilde X_t&= b^X(\widetilde X_t) \ud t + \sigma^X(\widetilde X_t)\ud W^1_t, \quad \widetilde X_0 =x_0\in \R.
	\end{align}
It is easy to see that $\widetilde{X}$ has the same dynamics as the ``original'' factor $X$ except for the jump at $\tau_R$.  We define $\widetilde N_t:=\eta_{\widetilde \theta_t}$ for every $t \geq 0$, where $\widetilde \theta_t=\int_0^t\lambda^L(\widetilde X_s)\ud s$, then  $\widetilde N=(\widetilde N_t)_{t \geq 0}$ is a  doubly stochastic point process with intensity $(\lambda^L(\widetilde X_{t}))_{t \geq 0}$. We can use these processes to construct $\widetilde L=(\widetilde L_t)_{t \geq 0}$ as follows,
	\begin{align}
	\widetilde L_t=M_{\widetilde \theta_t}, \quad t \geq 0.
	\end{align}
	Notice that  before default,  the triples $(X,N,L)$ and $(\widetilde X, \widetilde N,\widetilde L)$  coincide, that is the processes
	$(1-H_t^R) (X_t,N_t,L_t) $ and $(1-H_t^R) (\widetilde X_t, \widetilde N_t,\widetilde L_t)$ are indistinguishable.
	
We let $\bF:=\bF^W\vee\bF^{\widetilde L}$. The following result holds.
\begin{lemma}\label{lemma:doubly_stochastic}
		The random time $\tau_R$ is doubly stochastic with respect to the background filtration  $\bF$.	\end{lemma}
	\begin{proof}
By the construction  of  $\tau_R$ we have $\Q( \tau_R > t \mid \F^W_\infty \vee\F^{M}_\infty) =  e^{-\int_0^t  \lambda^R( Y_s) \ud s}$ for every $t \geq 0$.
		Now we observe that $\lambda^R(Y)$ is adapted to $\bF^W$ and so is
		$(e^{-\int_0^t  \lambda^R( Y_s) \ud s})_{t \geq 0}$. Moreover we have that  $$\F^W_\infty \vee\F^{\widetilde{L}}_\infty\subseteq \F^W_\infty \vee\F^M_\infty\,, $$ which implies that
		$\Q( \tau_R > t \mid \F^W_\infty \vee\F^{\widetilde{L}}_\infty) =  e^{-\int_0^t  \lambda^R( Y_s) \ud s}.$
\end{proof}

\section{Credit Value Adjustment}\label{sec:CVA}
	
	To resume the problem, we  consider a reinsurance contract between $I$ and $R$ with maturity $T$ and  payoff $\phi(L_T)$ for a nonnegative and  increasing  function $\phi$. For technical reasons we assume that $\phi$ is bounded and Lipschitz continuous; this assumption holds for the examples considered in Section~\ref{subsec:setup}. Moreover,  the counterparty-risk free  market value of this contract is given by
	$$V_t^\phi:=\esp{e^{-r(T-t)}\phi(L_T)|\G_t}\,,\quad 0 \le t \le T.$$
	We assume that the premium for the contract has been paid at $t=0$ so that $I$ has no financial obligation towards $R$.
	If $R$ defaults before the maturity date $T$,  the insurance company  needs to renew her protection, that is she needs to buy a new reinsurance contract at the post-default market value $V^\phi_{\tau_R}$. We assume that  $I$ receives a recovery payment of size $(1-\delta^R)V^\phi_{\tau_R}$  where $\delta^R \in (0,1]$   is the \emph{loss  given default} of $R$. Hence $I$ suffers a loss of size  $\delta^R V^\phi_{\tau_R}$. We denote by $\CL=(\CL_t)_{0 \le t \le T}$
the  payment stream arising from the  counterparty-risk loss. We have that
	\begin{align}\label{eq:CL}
	\CL_t:=\delta^R V_{\tau_R}^\phi \I_{\{\tau_R\leq t\}}=\delta^R\int_0^t V_s^\phi \ud H_s^R, \quad 0\le t \le T.
	\end{align}

	Note that under wrong-way risk, i.e. with  $\esp{V^{\phi}_t|\tau_R=t}>\esp{V^{\phi}_t}$, the loss of $I$ at $\tau_R$ is higher than its unconditional value. This is an important issue in the management of RCCR. For instance, in the Solvency~II regulation  it is stated that
	``As the failure of the counterparty  is more likely when the potential loss is high, the LGD [in our case the  loss caused by the default of $R$] should be determined for the case of a stressed situation,'' see \citet{bib:ceiops-09b}. It is a strong point of our approach that wrong-way risk is generated endogenously by the model. In contrast, in the standard formula of Solvency~II ad-hoc adjustments are necessary to account for wrong-way risk.

	We define the \emph{credit value adjustment} (CVA) for the reinsurance contract as the market consistent value of the future credit loss, that is
	\begin{align}\label{eq:def-CVA}
	\CVA_t=  \esp{\int_t^T \delta^R V_s^\phi e^{-r(s-t)} \ud H_s^R | \G_t},\quad  0 \le t \le T.
	\end{align}
	The amount $\CVA_t$ can be viewed as a risk reserve that the insurance company has to set aside at time $t$ to cover for losses due to reinsurance counterparty risk.		Alternatively, $\CVA_{t_0}$ can be viewed as the \emph{pricing adjustment} to account for RCCR at time $t_0$, that is on  $\{\tau_R>t_0\}$   the market consistent value  of the cash-flow that is actually received  by $I$ is equal to $V_{t_0}^\phi - \CVA_{t_0}$. This follows from  the following lemma.
	\begin{lemma} \label{lemma:CVA} For $0\le t_0\le T$ one has
		$$ \esp{\int_{t_0}^T e^{- r (s-t_0)}V_s^\phi \ud H^R_s \mid \G_{t_0} } = \I_{\{\tau_R>t_0\} } \esp{ H_T^R e^{-r(T-t_0)} \phi(L_T) \mid \G_{t_0} }.$$
	\end{lemma}
	\begin{proof}
		Define the stopping time $\sigma_R :=  (\tau_R \wedge T) \vee t_0$. Since $(e^{-rt} V_t^\phi)_{0 \le t \le T}$  is a $(\bG, \Q)$-martingale and   $\sigma_R \le T$, we get from the optional sampling theorem that
		\begin{equation}\label{eq:opt-sampl}
		V_{\sigma_R}^\phi = \esp{ e^{- r(T- \sigma_R)} \phi(L_T) \mid \G_{\sigma_R} }.
		\end{equation}
		Notice that $\sigma_R=\tau_R$ on the set $\{t_0 < \tau_R\leq T\}$ and  therefore using equation \eqref{eq:opt-sampl} we get
		\begin{align*}
		&\esp{\int_{t_0}^T e^{- r (s-t_0)}V_s^\phi \ud H^R_s \mid \G_{t_0} } =
		\esp{\I_{\{t_0 < \tau_R\leq T\}}  e^{-r(\tau_R -t_0)}  V_{\tau_R}^\phi  \mid \G_{t_0} } \\
		&\quad = \esp{\I_{\{t_0 < \tau_R\leq T\}}  e^{-r(\sigma_R -t_0)}  V_{\sigma_R}^\phi  \mid \G_{t_0} }  =     \esp{ \esp{ \I_{\{t_0<\tau_R\leq T\}} e^{- r(T- {t_0})}   \phi(L_T) \mid \G_{\sigma_R} }  \mid  \G_{t_0}},
		\end{align*}
		so that the lemma follows from iterated conditional expectations (as $\G_{t_0}\subseteq\G_{\sigma_R} $).
	\end{proof}
Now we return to the interpretation of  the $\CVA$. Fix $t_0 \in [0, T]$. On $\{\tau_R >t_0\}$ the  cash flow actually received by $I$ is given by
	$ \phi(L_T) (1-H_T^R) + (1- \delta^R)\int_{t_0}^T V_s^\phi \ud H^R_s  \,.$
	The expected discounted value of this cash-flow equals
	$$  V_{t_0}^\phi  - \esp{e^{-r(T-{t_0})} \phi(L_T) H_T^R \mid \G_{t_0}} + \esp{ \int_{t_0}^T e^{- r (s-{t_0})}V_s^\phi \ud H^R_s \mid \G_{t_0} } - \CVA_{t_0}
	$$
	which is equal to $V_{t_0}^\phi - \CVA_{t_0} $, as the terms in the middle cancel by Lemma~\ref{lemma:CVA}.

	Next we want to represent the value of the CVA as classic solution of a  partial integro-differential equation (PIDE). This allows for an alternative characterization of the adjusted price in addition to the stochastic representation given in equation \eqref{eq:def-CVA}, and it is essential for the computation of the hedging strategy in Section \ref{sec:hedging}.
	As a first step we analyze the term $V_{\tau_R}^\phi$ that appears in the definition of the credit loss.
	Note that  the shifted process $(X_{\tau_R+t},L_{\tau_R +t })_{t \ge 0}$ has the same dynamics as $(\widetilde X_t, \widetilde L_t)_{t \ge 0}$; hence it is a two-dimensional Markov process with generator
	\begin{align}\label{eq:generator-tilde-processes}
	\mathcal L^{(\widetilde L,\widetilde X)}f(t,l,x)&=\frac{\partial f}{\partial x }(t,l,x) b^X(x)+\frac{1}{2}\frac{\partial^2 f}{\partial x^2 }(t,l,x)(\sigma^X(x))^2 \\
	&+ \int_{\R^+} \big(f(t, l+z,x)-f(t,l,x)\big) \lambda^L(x)\nu(\ud z).
	\end{align}
	This suggests that $V_{\tau_R}^\phi$ can be described as the solution of a backward equation involving the generator $\mathcal L^{(\widetilde L,\widetilde X)}$. The next proposition shows that this is in fact correct.

	\begin{proposition}\label{prop:PIDE1}
		Under Assumption \ref{ass:regularity}, there exists a unique bounded  classical solution $v^\phi$ (i.e. continuous, $\mathcal C^1$ in $t$ and $\mathcal C^2$ in $x$) of the following backward PIDE 
		\begin{align}\label{eq:backward1}
		&\frac{\partial v^\phi}{\partial t}(t,l, x)+ \mathcal L^{(\widetilde L, \widetilde X)}v^\phi(t,l,x)=r v^\phi(t,l, x),\quad (t,l, x)\in [0,T)\times \R^+\times \R,
		\end{align}
		with terminal  condition $v^\phi(T,l, x)=\phi(l)$.
		Moreover, it holds for $\tau_R  \le T$ that
		$$V^{\phi}_{\tau_R}= v^\phi\big (\tau_R, \widetilde L_{\tau_R}, \widetilde X_{\tau_R}+\gamma^X(\widetilde X_{\tau_R})\big )\,.$$
		
	\end{proposition}
	
	\begin{proof}
		The process $(\widetilde L, \widetilde X)$ is a two-dimensional Markov process with pure jump component $\widetilde L$ and generator $\mathcal L^{(\widetilde L, \widetilde X)}$ given in \eqref{eq:generator-tilde-processes}. The existence of  a classical solution $v^\phi$ to the backward equation~\eqref{eq:backward1} follows from \citet{bib:colaneri-frey-19}. Moreover,   it holds that
		$$v^\phi(t,l,x) = \esp{ e^{-r(T-t)} \phi(\widetilde L_{T}) \mid \widetilde L_t =l,\widetilde X_t =x }\,. $$
		The strong Markov property thus gives that on $\{\tau_R \le T\}$,
		$$
		V^{\phi}_{\tau_R}= v^\phi\big ( \tau_R, L_{\tau_R}, X_{\tau_R}\big ) =  v^\phi(\tau_R, \widetilde L_{\tau_R}, \widetilde X_{{\tau_R}}+\gamma^X(\widetilde X_{{\tau_R}}))\,,
		$$
		where in the last equality we used that $L_{\tau_R} = \widetilde{L}_{\tau_R}$,   $X_{\tau_R } = \widetilde X_{{\tau_R}}+\gamma^X(\widetilde X_{{\tau_R}})$\, and $\widetilde X_{{\tau_R}-}=\widetilde X_{{\tau_R}}$.
	\end{proof}
Note that the regularity properties of the  function $v^\phi$   ($\mathcal C^1$ in $t$, $\mathcal C^2$ in $x$  but  only continuous in $l$) are due to  the fact that $\widetilde L$ is a pure jump process and therefore the smoothing effect coming from the diffusion does not apply in the $l$ direction.  In the statement of Proposition~\ref{prop:PIDE1} we refer for brevity to Assumption \ref{ass:regularity}. However, Proposition~\ref{prop:PIDE1} does not involve the process $Y$ and therefore some of the conditions in the list (A1)--(A4) are unnecessary.
	
	\begin{proposition}\label{prop:CVA}
		Under Assumptions \ref{ass:regularity} the value of the $\CVA$ is given by
		\begin{align}\label{eq:CVAprice}
		{\CVA}_t=\delta^R (1-H^R_t) f^{\CVA}(t, L_t, X_t, Y_t)
		\end{align}
		where $f^{\CVA}:[0,T]\times \R^+\times \R\times \R\to \R^+$ is a classical solution (i.e. continuous, $\mathcal C^1$ in $t$ and $\mathcal C^2$ in $(x,y)$) of the following backward PIDE
		\begin{align}
		&\frac{\partial f^{\CVA}}{\partial t} + \mathcal L^{(\widetilde L, \widetilde{X},Y)}f^{\CVA} +\lambda^R(y) v^{\phi}(t,l,x+\gamma^X(x))=(\lambda^R(y)+r) f^{\CVA}, \label{eq:PIDE2}
		\end{align}
		for all $(t,l, x, y)\in [0,T)\times \R^+\times \R^2$ with  terminal  condition $f^{\CVA}(T,l, x,y)=0$.  The operator $\mathcal L^{(\widetilde L,\widetilde X,Y)}$ (the generator of the three-dimensional Markov process $(\widetilde L,\widetilde X,Y)$) is given by
		\begin{align}
		\mathcal L^{(\widetilde L,\widetilde X,Y)}f &=\frac{\partial f}{\partial x } b^X(x)+\frac{\partial f}{\partial y } b^Y(y)+\frac{1}{2}\frac{\partial^2 f}{\partial x^2 }(\sigma^X(x))^2
		+\frac{1}{2}\frac{\partial^2 f}{\partial y^2 }(\sigma^Y(y))^2 \\ &+\frac{\partial^2 f}{\partial x \partial y}\rho\sigma^X(x)\sigma^Y(y)
		+ \int_{\R^+} (f(t, l+z,x,y)-f(t,l,x,y)) \lambda^L(x)\nu(\ud z), \label{eq:generator2}
		\end{align}
		where $f$ is always evaluated at $(t,l,x,y)$.
	\end{proposition}

	\begin{proof}
		The $\CL$ is a so-called payment-at-default claim (see for instance \citet[Section 10.5]{bib:mcneil-frey-embrechts-15}). Proposition~\ref{prop:PIDE1} allows to express its payoff at $\tau_R$ in terms of contagion free quantities. Then  we get that
		\begin{align}
		\CVA_t 
		&=\esp{\int_t^T\delta^R v^\phi (s, \widetilde L_s, \widetilde X_{s}+\gamma^X(\widetilde X_{s}))e^{-r(s-t)}\ud H^R_s \mid \G_t}.\label{eq:equality0}
		\end{align}
		In  equation \eqref{eq:equality0} we can replace $\G_t$ with $\F_t \vee \H_t$ since these sigma fields coincide up to time $\tau_R$. Then we get from  Lemma \ref{lemma:doubly_stochastic}
		and ~\citet[Theorem~10.19]{bib:mcneil-frey-embrechts-15} that
		\begin{align}
		\CVA_t =\delta^R(1-H^R_t)\esp{\int_t^T v^\phi (s, \widetilde L_s, \widetilde X_{s}+\gamma^X(\widetilde X_{s}))\lambda^R(Y_s) e^{-\int_t^s(r+\lambda^R(Y_u))\ud u} \ud s \mid \F_t}. \label{eq:equality2}
		\end{align}
		Note that the process $(\widetilde L, \widetilde X, Y)$ is Markovian with respect to the filtration $\mathbb{F}$  with generator  $\mathcal L^{(\widetilde L,\widetilde X,Y)}$ as in \eqref{eq:generator2}. It follows that  there is a function $f^{\CVA}:[0,T]\times \R^+\times \R\times \R\to \R^+$ such that
		\begin{align}
		{\CVA}_t=\delta^R (1-H^R_t) f^{\CVA}(t,\widetilde L_t, \widetilde X_t, Y_t).
		\end{align}
		Then, by applying \cite[Theorem 1]{bib:colaneri-frey-19} we get that $f^{\CVA}$ is a classical solution of the backward PIDE \eqref{eq:PIDE2}.
		Finally note that on the event $\{\tau_R> t\}$, $1-H^R_t=1$ and also $\widetilde L_t=L_t$, $\widetilde X_t=X_t$, which implies \eqref{eq:CVAprice}.
	\end{proof}

\begin{example}\label{example1}
In the numerical analysis we  consider a special case of our setting. There the loss intensity $\lambda^L$ is constant except for an upward jump at time $\tau_R$ that models price contagion.
In this case we may identify the intensity $\lambda^L$ and the intensity-factor process $X$ (i.e. $\lambda^L(\cdot)$ is the identity function) and assume that
		\begin{align}\label{eq:lambdaL-example}
		\lambda^L (X_t)= X_t=x_0(1 +  H^R_t \gamma) , \quad 0 \le t \le T,
		\end{align}
		for  constants $x_0>0$ and $\gamma>0$. Here the parameter $\gamma$ models the percentage change in the loss intensity at $\tau_R$.  		We now calculate the credit value adjustment for this  situation.  Under \eqref{eq:lambdaL-example} the process ${\widetilde L}$ is  a compound Poisson process with  intensity $x_0$,  jump-size distribution $\nu(\ud z)$ and generator
\begin{align}
		\mathcal L^{\widetilde L}_{x_0} f(t,l)=x_0 \int_{\R^+} \big(f(t, l+z)-f(t,l)\big) \nu(\ud z). \label{eq:generator1example}
\end{align}
For given $x_0>0$, define the function
		$(t,l) \mapsto v^{\phi}(x_0;t,l)$ as the solution of the backward integral equation
		\begin{align}\label{eq:backward2}
		&\frac{\partial v^\phi}{\partial t}(x_0;t,l)+ \mathcal L^{\widetilde L}v^\phi(x_0;t,l)=r v^\phi(x_0;t,l),\quad (t,l)\in [0,T)\times \R^+,
		\end{align}
		with terminal  condition $v^\phi(x_0;T,l)=\phi(l)$.		Then, the post default value of the reinsurance contract is given by\footnote{Of course other actuarial techniques such as Panjer recursion could be used as well to compute $ v^{\phi}$.}
		$$V^{\phi}_{\tau_R}= v^\phi (x_0(1+\gamma); \tau_R, \widetilde L_{\tau_R} ).$$	
		With this we get that credit value adjustment satisfies ${\CVA}_t=\delta^R (1-H^R_t) f^{\CVA}(x_0; t, {\widetilde L}_t, Y_t)$, where
		the function $(t,l,y) \mapsto f^{\CVA}(x_0;t,l,y)$ is the solution of the backward PIDE
		\begin{align}
		&\frac{\partial f^{\CVA}}{\partial t}(x_0;t,l,y)+ \mathcal L_{x_0}^{(\widetilde L, Y)}f^{\CVA}(x_0;t,l,y)+\lambda^R(y) v^{\phi}(x_0(1+\gamma);t,l)=(\lambda^R(y)+r) f^{\CVA}(x_0;t,l,y), \qquad \label{eq:PIDE2EX}
		\end{align}
		for every $(t,l, y)\in [0,T)\times \R^+\times \R$ with  terminal  condition $f^{\CVA}(x_0;T,l,y)=0$, and where for a generic continuous function $f(l,y)$ which is $\mathcal C^2$ in $y$, the operator $\mathcal L^{(\widetilde L,Y)}_{x_0}$ is given by
		\begin{align}
		&\mathcal L^{(\widetilde L,Y)}_{x_0}f(l,y)=\frac{\partial f}{\partial y }(l,y) b^Y(y)+\frac{1}{2}\frac{\partial^2 f}{\partial y^2 }(l,y)(\sigma^Y(y))^2 + x_0 \int_{\R^+} (f(l+z,y)-f(l,y)) \nu(\ud z). \label{eq:generator2example}
		\end{align}
Note that in this example the variable corresponding to loss intensity drops out of the equation \eqref{eq:PIDE2EX} and therefore (A2) in Assumption \ref{ass:regularity} can be replaced by the simpler condition
		\begin{itemize}
			\item[(A2')] There is some $\beta >0$ such that $\sigma^Y(\cdot)>\beta$.
		\end{itemize}
\end{example}

	\section{Hedging of Reinsurance Counterparty Credit Risk}\label{sec:hedging}
	
	In this section we investigate how the insurance company can reduce the losses arising from the default of the reinsurer by a dynamically adjusted position in a  credit default swap (CDS) on $R$. A CDS is a natural hedging instrument for credit risk since it makes  a payment at $\tau_R$, that is exactly when the counterparty risk loss arises. Moreover, there is a reasonably liquid market for CDSs on major reinsurane companies. Another option for managing counterparty risk would be a dynamically adjusted collateralization strategy as in~\citet{bib:frey-roesler-14};  however, one of the advantages of hedging with CDS contracts is that  a strategy can be implemented unilaterally by $I$.
	In our setting  there are  several sources of randomness that do not correspond to traded assets, such as the loss process $L$ or the loss intensity $\lambda^L$, and therefore perfect hedging  is not possible.  To deal with the ensuing {market incompleteness} we resort to a quadratic hedging method. Precisely we will consider self financing strategies and minimize the {quadratic hedging error} at the maturity date.

To proceed with a formal analysis of the hedging problem we need to discuss the dynamics of a self-financing CDS trading strategy. This issue is taken up next.
	\subsection{Dynamics of a CDS trading strategy}\label{subsec:CDS}
	We consider a CDS contract on $R$  with fixed running spread premium  $\zeta>0$ and with default payment given by  the deterministic loss given default $\delta^{\CDS} \in (0,1]$ of $R$.  To simplify the exposition we  assume that the premium payments are made continuously. The cashflow stream associated to the CDS (from the viewpoint of $I$) is therefore given by
	\begin{equation}\label{div}
	D^R_t =\delta^{\CDS} H^R_t -\zeta \int_0^t (1-H_u^R) \ud u, \quad 0 \le t \le T,
	\end{equation}
	where the first term refers to the 
	payment at default and the second term is the premium payment. Note that \eqref{div} describes the cash-flows of a CDS contract with notional equal to one; holding $m$ units of this contract is the same as holding one CDS contract with notional $m$.
	
	The present value of the future payments of the CDS  is given by
	\begin{align}
	\Lambda_t& :=\esp{\int_t^T e^{-r(u-t)}\ud D^R_u|\G_t}=\esp{\delta^{\CDS}\int_t^T e^{-r(u-t)}  \ud H^R_u -\zeta\int_t^T e^{-r(u-t)} (1-H^R_u) \ud u|\G_t}.\end{align}
	Similarly as in Section \ref{sec:CVA}, we characterize the process $\Lambda$ in terms of the classical solution of a backward partial differential equation (PDE).
	
	\begin{proposition}\label{prop:Lambda}
		Under Assumptions \ref{ass:regularity} the process $\Lambda$ is given by
		$$
		\Lambda_t = (1-H^R_t) g(t, Y_t)
		$$
		where $g:[0,T]\times \R\to \R$ is a classical solution (i.e.  $\mathcal C^1$ in $t$ and $\mathcal C^2$ in $y$) of the following backward PDE
		\begin{align}
		&\frac{\partial g}{\partial t}(t,y)+ \mathcal L^{Y}g(t,y)+(\delta^{\CDS}\lambda^R(y)-\zeta)=(\lambda^R(y)+r) g(t,y),
      \quad (t,y)\in [0,T)\times \R, \label{eq:PIDE3}
		\end{align}
		 with terminal  condition $g(T,y)=0$. Here the operator $\mathcal L^{Y}$ is the generator of $Y$, that is
		\begin{align}
		\mathcal L^{Y} f(y)&=\frac{\partial f}{\partial y }(y) b^Y(y)+\frac{1}{2}\frac{\partial^2 f}{\partial y^2 }(y)(\sigma^Y(y))^2. \label{eq:generator3}
		\end{align}
	\end{proposition}
	
	\begin{proof}
		Since $M^R$ in \eqref{mg2} is a $\bG$- martingale we have that
		\begin{align}
		\Lambda_t&=\esp{\int_t^Te^{-r(u-t)} ( \delta^{\CDS} \lambda^R(Y_u) - \zeta) (1-H^R_u) \ud u|\G_t  }\label{eq:CDSv_0}
		\end{align}
		Using Fubini's theorem,  Lemma \ref{lemma:doubly_stochastic} and~\citet[Theorem~10.19]{bib:mcneil-frey-embrechts-15}
		we get that the right hand side of \eqref{eq:CDSv_0} is equal to
		\begin{align}
		(1-H^R_t)\esp{\int_t^Te^{-\int_t^u(r+\lambda^R(Y_s))\ud s} (\delta^{\CDS} \lambda^R(Y_u) - \zeta)  \ud u|\F_t  }\label{eq:CDS-value}
		\end{align}
		By Markovianity of the process $Y$ with respect to filtration $\bF$, there exists a function $g$ such that conditional expectation in \eqref{eq:CDS-value} is equal to $g(t,Y_t)$. Denote by $\mathcal L^{Y}$  the generator of  $Y$ given by  \eqref{eq:generator3}. Then it is easily seen that under Assumption~\ref{ass:regularity}, $g$ is the
		classical solution of \eqref{eq:PIDE3}, see, e.g. \citet[Theorem 8.2.1]{oksendal-book}.
	\end{proof}
	
	Finally we define the \emph{discounted gains  process} of the CDS (the past cashflows and the present value of the future cashflows, both discounted back to time zero)  by
	\begin{align}\label{eq:gainCDS}
	S_t= e^{-rt} \Lambda_t+ \int_0^t e^{-ru}\ud D^R_u, \quad 0 \le t \le T.
	\end{align}
	Note that $S_t= \esp{\int_0^T e^{-ru}\ud D^R_u |\G_t }$ for every $0 \le t \le T$ and therefore $S$ is a square integrable $(\bG,\Q)$-martingale.
	Consider now a self-financing trading strategy $\xi=(\xi^0, \xi^1)$, where $\xi^1_t$ is the notional of the CDS position at time $t$ and where $\xi^0_t$ is the cash position at time $t$. Then the value of this strategy at time $0 \le t \le T$ equals $V_t(\xi) = \xi^1_t \Lambda_t + \xi^0_t e^{-r t}$, and the strategy is \emph{selffinancing} if  the discounted value $\widetilde{V}_t(\xi) = e^{-r t} V_t(\xi)$ satisfies
$$\widetilde{V}_t(\xi) = V_0(\xi) + \int_0^t \xi^1_s \ud S_s\,,\quad 0\leq t \leq T\,. $$

	\subsection{Quadratic hedging}\label{subsec:hedging}

	Next we formalize the quadratic criterion that is used to determine the optimal hedging strategy. We call a self-financing trading strategy $\xi=(\xi^0, \xi^1)$ \emph{admissible} if $\xi^0$ is $\bG$-adapted and  $\xi^1$ is $\bG$-predictable and  satisfies the  integrability condition
	\begin{equation}\label{eq:admissible}
	\esp{  \int_0^T (\xi^1_u)^2  \ud\langle S\rangle_u}<\infty\,.
	\end{equation}
	Here  $\langle S\rangle$ denotes the \emph{predictable quadratic variation} of the martingale $S$ (the predictable compensator of the pathwise quadratic variation $[S]$ of $S$). Condition \eqref{eq:admissible} ensures that the discounted value process $V(\xi)$ is a right continuous and square integrable martingale. 	The hedging problem amounts to finding a self-financing  admissible strategy $\xi^*$ with initial value $V_0(\xi^*)$ and CDS position $\xi^{1,*}$ that minimizes
	the quadratic hedging error
	\begin{align}\label{eq:criterion}
	\esp{\left(\int_0^Te^{-rt}\delta^R V^{\phi}_t\ud H^R_t- \Big(V_0(\xi)+\int_0^T\xi^1_t  \ud S_t\Big)\right)^2}.
	\end{align}
	Such a strategy will be called \emph{$\Q$-mean-variance minimizing.}
	
	\begin{remark} We continue with a few comments on the hedging criterion.

1) Minimizing the quadratic hedging error with respect to  the risk-neutral measure $\Q$, instead of the historical measure $\P$, has a couple of advantages.  First, the ensuing CDS position $\xi^{1,*}$ is time-consistent: the  CDS strategy that minimizes the conditional quadratic hedging error
		$$\esp{\left(\int_t^T e^{-rs}\delta^R V^{\phi}_s\ud H^R_s- \Big ( V_t(\xi)+\int_t^T\xi^1_s  \ud S_s \Big)\right)^2\Big|  \ \G_t}$$
		over the period $[t, T]$ is the restriction of  $\xi^{1,*}$ to the interval $[t,T]$. This is in general not true for a $\P$-mean-variance minimizing strategy.
		Moreover, since the default  and loss intensities under $\Q$ are typically higher than the corresponding $\P$-intensities  (see Remark~\ref{rem:market-consistent-valuation}),  more mass is put in expectation \eqref{eq:criterion} on states  where the counterparty-risk loss  is large and the $\Q$-mean-variance minimizing  strategy will track the credit loss  more closely in those states than a $\P$-mean-variance-minimizing  strategy; this adds an additional layer of prudence to our approach. Finally a $\Q$-mean-variance-minimizing strategy is comparatively easy to determine and the solution has a clear economic interpretation.
		
2) As an alternative to $\Q$-mean-variance minimization  one might consider  \emph{risk minimization} under $\Q$ as hedging criterion. The investment in the risky asset (the CDS in our setting) is  the same for both approaches; the only difference is that in the mean-variance-hedging approach  a self financing strategy is followed until time $T$ where the hedging error  takes the form of a lump sum adjustment.  In the risk minimization approach on the other hand   the portfolio value  is adjusted  continuously at any  $0 < t \le T$. Note however that mean-variance hedging and risk minimization lead to  different strategies if one works under the historical measure.  For an in-depth discussion of these issues we refer to \citet{schweizer2001guided}.
	\end{remark}

To determine the $\Q$-mean-variance minimizing strategy we first introduce the discounted gain process $M^{\CL}$ associated with the credit loss. This process is given by
	\begin{align}\label{eq:MCL0}
	M^{\CL}_t&=\esp{\int_0^T e^{-rs}\ud\CL_s | \G_t} = \int_0^t e^{-rs}\ud\CL_s +  e^{-rt}\CVA_t, \quad 0 \le t \le T,
	\end{align}
	where we recall that CL is the payment stream arising from the counterparty-risk loss defined in equation~\eqref{eq:CL}, and $M^{\CL}$ is easily seen to be  a square integrable $(\bG, \Q)$-martingale. Since the discounted gain process of the CDS in equation \eqref{eq:gainCDS} is a $(\bG, \Q)$-martingale, it is well known that the $\Q$-mean-variance optimal strategy can  be determined  with the help of the Galtchouk-Kunita-Watanabe decomposition of  $M^{\text{CL}}$ with respect to  $S$. This result ensures the existence of a predictable process $\xi^{1,*}$ satisfying \eqref{eq:admissible}  and of a martingale $ A$ null at time zero, which is strongly orthogonal to $ S$ (that is the product of the two martingales $(S_t A_t)_{0 \le t \le T}$ is also a martingale or, equivalently, the predictable quadratic covariation $\langle S, A\rangle $ vanishes) such that
	\begin{equation}\label{eq:GKW_M}
	M^{\CL}_t= M^{\CL}_0 + \int_0^t \xi_u^{1,*} \ud S_u + A_t,\quad \Q-a.s. \quad 0 \le t \le T.
	\end{equation}
	Then the strategy $\xi^*$ with CDS position $\xi^{1,*}$ and initial value $V_0(\xi^*) = M_0^{\CL}$ is $\Q$-mean-variance minimizing.  A detailed proof of this result can be found in~\citet{schweizer2001guided}.
	Intuitively, decomposition~\eqref{eq:GKW_M}  permits to decompose the payment stream $\CL$ into its attainable part given by  $\int \xi_t^{1,*} \ud S_t$,   and an  unattainable part $A$ corresponding to non-hedgeable risk.
	
	Identifying  $\xi^{1,*}$  entails taking the predictable covariation with respect to $ S$ on both sides of equation \eqref{eq:GKW_M}. Using orthogonality between $A$ and $S$, we get that
	\[
	\langle M^{\CL}, S\rangle_t=\int_0^t \xi_u^{1,*}  \ud \langle S \rangle_u, \quad 0 \le t \le T,
	\]
	where $\langle M^{\CL},  S\rangle$ denotes the predictable quadratic covariation  between martingales $M^{\CL}$ and   $ S$.  This implies that
	$\xi^{1,*}$ can be identified as predictable version of the Radon Nikodym density $\frac{\ud \langle M^{\CL}, S\rangle}{\ud \langle S\rangle}$. Computing this density  is the key point in the proof of the following theorem where  we determine the $\Q$-mean-variance  minimizing strategy.

	\begin{theorem}\label{thm:strategy}The $\Q$-mean-variance minimizing  strategy is characterized by the initial value $V_0(\xi^*) = \CVA_0$ and by the CDS position
		$ \xi^{1,*}_t =  \displaystyle{\frac{\nicefrac{\ud \langle M^{\text{CL}}, S \rangle_t}{\ud t}} { \nicefrac{\ud \langle S\rangle_t}{\ud t}}}$, for every $0 \le t \le T$,
		where
		\begin{align}
		\frac{\ud \langle M^{\CL}, S\rangle_t}{\ud t}  = & \ \delta^Re^{-2rt}(1-H^R_{t-})\bigg \{
		\rho  \sigma^X(X_{t-}) \sigma^Y(Y_t) \frac{\partial f^{\CVA}}{\partial x }(t, L_{t-}, X_{t-}, Y_t) \frac{\partial g}{\partial y }(t, Y_t) \label{eq:dMdS}\\
		& +  \!  (\sigma^Y(Y_t))^2 \  \frac{\partial f^{\CVA}}{\partial y }(t, L_{t-}, X_{t-}, Y_t) \frac{\partial g}{\partial y }(t, Y_t)  \\
		&  + \! \lambda^R(Y_t) \big(\delta^{\CDS}\! -\! g(t, Y_{t})\big ) \big (v^\phi(t, L_{t-}, X_{t-} \!+\! \gamma^X (X_{t-})) - f^{\CVA}(t, L_{t-}, X_{t-}, Y_t) \big)     \bigg \}
		\intertext{and}\label{eq:dSdS}
		\frac{\ud \langle S\rangle_t}{\ud t}  =& \ e^{-2rt}(1-H^R_{t-})\bigg\{\lambda^R(Y_t)(\delta^{\CDS} - g(t, Y_{t}))^2 +(\sigma^Y(Y_t))^2   \left(\frac{\partial g}{\partial y }(t, Y_t)\right)^2 \bigg\} \,.
		\end{align}
	\end{theorem}

	\begin{proof} By definition $M_0^{\CL} = \CVA_0$ which gives the initial value of the strategy. In order to determine $\xi^{1,*}$ note  that in our setting the processes $\langle M^{\CL}, S\rangle$ and $\langle S\rangle$ are absolutely continuous with respect to Lebesgue measure.
		This implies that $\Q$-a.s.
		$$  \frac{\ud \langle M^{\CL}, S\rangle_t}{\ud \langle S\rangle_t} = \frac{\nicefrac{\ud \langle M^{\text{CL}}, S \rangle_t}{\ud t}} { \nicefrac{\ud \langle S\rangle_t}{\ud t}}\,, \quad  0 \le t \le T.
		$$
		To derive   the processes $ \frac{\ud \langle M^{\CL}, S\rangle_s}{\ud s} $ and $  \frac{\ud \langle S\rangle_s}{\ud s}$ we  compute the pathwise quadratic (co)variations  $[M^{\CL}, S]$, respectively  $[S]$, and we use that  $\langle M^{\CL}, S \rangle $, respectively  $\langle S\rangle $, is  the  predictable compensator of these processes.
		We recall that $M^R$ is the compensated martingale given in equation \eqref{mg2} and denote by $\widetilde m(\ud t, \ud z)$ the compensated jump measure $\widetilde m(\ud t, d z)=m^L(\ud t, \ud z)-\lambda^L(X_{t-})\nu(\ud z)$. From the  PIDE characterization of the CVA in Proposition \ref{prop:CVA} and the It\^o formula, see Appendix \ref{Appendix1} for the detailed computations, we get that the martingale $M^{\CL}$ in \eqref{eq:MCL0} is explicitly given by
		\begin{align}
		M^{\CL}_t& =M^{\CL}_0+\delta^R\int_0^te^{-rs}(v^\phi(s, L_{s-}, X_{s-} + \gamma^X(X_{s-}))-f^{\CVA}(s, L_{s-}, X_{s-}, Y_s))\ud M^R_s \label{eq:MCL}\\
		&+\delta^R\!\int_0^t\!\!e^{-rs}(1\!-\!H^R_{s-}) \Big(\sigma^X(X_{s-})\frac{\partial f^{\CVA}}{\partial x }(s, L_{s-}, X_{s-}, Y_s) \!+\!\rho \sigma^Y(Y_s)\frac{\partial f^{\CVA}}{\partial y }(s, L_{s-}, X_{s-}, Y_s) \Big) \ \ud W^1_s \quad {}\\
		&+\delta^R \!\int_0^t\!\!e^{-rs}(1\!-\!H^R_{s-}) \sigma^Y(Y_s) \frac{\partial f^{\CVA}}{\partial y }(s, L_{s-}, X_{s-}, Y_s)  \sqrt{1-\rho^2} \ \ud W^2_s \\
		&+\delta^R \!\int_0^t\!\!e^{-rs}(1\!-\!H^R_{s-}) \int_{\R^+} \Big(f^{\CVA}(s, L_{s-}\!+\!z, X_{s-}, Y_{s})-f^{\CVA}(s, L_{s-}, X_{s-}, Y_{s})\Big)
		\widetilde m(\ud s, \ud z),
		\end{align}
In a similar way we obtain the martingale decomposition  of the process $S$. It holds that for every $0\leq t \leq T$,
		\begin{align}
		S_t&=S_0+\int_0^te^{-rs}(\delta^{\CDS}-g(s, Y_{s}))\ud M^R_s\\
		&+\int_0^te^{-rs}(1-H^R_{s-}) \sigma^Y(Y_s) \frac{\partial g}{\partial y }(s, Y_s)  \left(\rho \ud W^1_s+ \sqrt{1-\rho^2} \ud W^2_s\right).\label{eq:CDS}
		\end{align}
		Then the quadratic covariation of the two martingales $M^{\CL}$ and $S$ and for the quadratic variation of $S$ is
		\begin{align}
		\ud [M^{\CL}, S]_t&=\delta^R e^{-2rt}\big(\delta^{\CDS}-g(t, Y_{t})\big)\big(v^\phi(t, L_{t-}, X_{t-} + \gamma^X(X_{t-}))-f^{\CVA}(t, L_{t-}, X_{t-}, Y_t)\big) \ud H^R_t \\
		&+\delta^Re^{-2rt} (1-H^R_{t-})\rho \ \sigma^X(X_{t-}) \sigma^Y(Y_t)\frac{\partial f^{\CVA}}{\partial x }(t, L_{t-}, X_{t-}, Y_t) \frac{\partial g}{\partial y }(t, Y_t) \ud t\\
		&+\delta^Re^{-2rt} (1-H^R_{t-})(\sigma^Y(Y_t))^2 \frac{\partial f^{\CVA}}{\partial y }(t, L_{t-}, X_{t-}, Y_t) \frac{\partial g}{\partial y }(t, Y_t) \ud t,\\
		\ud [S]_t=& \ e^{-2rt}(\delta^{\CDS}-g(t, Y_{t}))^2 \ud H^R_t + e^{-2rt} (1-H^R_{t-})(\sigma^Y(Y_t))^2 \left(\frac{\partial g}{\partial y }(t, Y_t)\right)^2 \ud t.
		\end{align}
		
		The predictable quadratic variation is then obtained by computing predictable compensators, which leads to \eqref{eq:dMdS} and \eqref{eq:dSdS} and implies the result.
		
	\end{proof}
	
	\subsubsection{Special cases and interpretation.} In order to understand the form of $\xi^{1,*}$ it is instructive to consider first the limiting case where $\sigma^X = \sigma^Y =0$ and where $\lambda^L_t= X_t= x_0 (1+ H_t^R \gamma) $ and $\lambda^R_t= \lambda^R (y_0)>0$ for every $0 \leq t \leq T$. In that setting we can consider both $x_0$ and $y_0$ as parameters and get that
	$$
	\xi^{1,*}_t= (1- H_{t-}^R)\frac{\delta^R \big(v^\phi(x_0 (1+ \gamma) ;t, L_{t-})- f^{\CVA}(x_0, y_0; t,  L_{t-})\big)}{\delta^{\CDS} - g(t, y_0)},  \quad  0\le t \le T.
	$$
	We clearly see that the CDS strategy generates  a payment at the default time  $\tau_R$ of size $\delta^R(v^\phi(x_0 (1+ \gamma) ;t, L_{\tau_R})-f^{\CVA}(x_0, y_0;\tau_R, L_{\tau_R}))$, that is the strategy provides a perfect hedge against the counterparty-risk loss   at $\tau_R$ (pure hedging of jump risk). Note however, that the strategy is not self-financing, as the CDS position needs to be adjusted according to the  random evolution of the aggregate claim amount  $L$.
	
For $\sigma^Y >0$  the strategy balances the hedging of jump risk and the hedging against fluctuations in the default intensity factor $Y$ (hedging of spread risk). The optimal mean-variance  strategy in the setting of Example \ref{example1} can be obtained by letting $\sigma^X =0$. Using the special notation for this case we obtain that
	{\small
		\begin{align}
		\xi^{1,*}_t= (1- H_{t-}^R) &\frac{\delta^R \lambda^R(Y_t) \big(\delta^{\CDS} - g(t, Y_{t})\big)  \big( v^\phi(x_0(1 + \gamma) ;t, L_{t-})-f^{\CVA}(x_0 ;t, L_{t-}, Y_t)\big)
		}{\lambda^R(Y_t) (\delta^{\CDS} - g(t, Y_{t}))^2 +(\sigma^Y(Y_t))^2  \left(\frac{\partial g}{\partial y }(t, Y_t)\right)^2 }\\
		&+ (1- H_{t-}^R) \frac{ \delta^R (\sigma^Y(Y_t))^2 \frac{\partial f^{\CVA}}{\partial y }(x_0;t, L_{t-}, Y_t)   \frac{\partial g}{\partial y }(t, Y_t)
		}{\lambda^R(Y_t)(\delta^{\CDS} - g(t, Y_{t}))^2 + (\sigma^Y(Y_t))^2 \left(\frac{\partial g}{\partial y }(t, Y_t)\right)^2 } \,.
		\end{align}}
	If $\sigma^X(\cdot)$, $\sigma^Y(\cdot) $ and $\rho$ are all strictly  positive, then an additional cross term
	$\rho  \sigma^X \sigma^Y \frac{\partial f^{\CVA}}{\partial x } \frac{\partial g}{\partial y }$ appears in \eqref{eq:dMdS}.   It is intuitively clear that both partial derivatives are positive\footnote{A higher loss intensity makes a large credit loss more likely, thereby increasing the CVA, and a higher default intensity increases the value  of the future CDS payments.}, so that the CDS position $\xi^{1,*}$ is increased by this term. This is due to the fact that  some of the risk caused by fluctuations in the non-traded loss intensity factor $X$ can be hedged by increasing the position in the correlated CDS contract.

	\section{Numerical Experiments}\label{sec:numerics}
	In this section we present results from numerical experiments that  complement the theoretical analysis. In Section \ref{sec:numericsCVAvalue} we focus on the relative importance of dependence and pricing contagion for wrong way risk; in Section \ref{sec:performance} we study   $\Q$-mean-variance-minimizing  strategies  and we   compare their performance to that of a static strategy.
	
Throughout our analysis we consider the following setup. We  identify processes the $X, Y$ and $\lambda^L, \lambda^R$, that is we assume that $\lambda^L(\cdot)$ and $\lambda^R(\cdot)$ are the identity functions.
 	The default intensity follows a CIR process with the dynamics
	\begin{align}
	\ud Y_t&=(0.05-Y_t) \ud t + 0.1\sqrt{Y_t}(\rho \ud W_t^1 + \sqrt{1-\rho^2}\ud W^2_t),\quad  Y_0=0.05;
	\end{align}
	this allows for an explicit formula for the price of the CDS, see, e.g. \citet{duffie2000transform}. For the loss intensity we consider a jump diffusion of the form
	\begin{align}
	\ud X_t&=  \gamma X_{t-}  \  \ud H^R_t + \kappa (100-X_t) \ud t + \sigma X_t \ud W^1_t, \quad X_0\in \R^+.
	\end{align}
	If we take $\kappa=\sigma=0$ we recover the case of Example \ref{example1} where the loss intensity has a jump at default and  is otherwise constant. Finally, we assume that claim sizes are Gamma$(\alpha, \beta)$ distributed.
		We consider a reinsurance contract of stop loss type with payoff $\phi(L_T)=[L_T-90]^+$, capped at $200$,  we set the interest rate to $r=0$ and  the loss-given-default of R and of the  CDS to $\delta^R=\delta^{\CDS}=1$.
	
	Next we briefly discuss the  methods used in the numerical analysis.
	The main task is to calculate the CVA in \eqref{eq:CVAprice}. Using the equivalent formulation  in \eqref{eq:equality2} we see that this amounts to evaluating the expectation
	\begin{align}
	\esp{\int_t^T v^\phi (s, \widetilde L_s, \widetilde X_{s}+\gamma \widetilde X_{s})\, Y_s \, e^{-\int_t^sY_u\ud u} \ud s \mid \ \widetilde L_t=l, \widetilde X_{t}=x, Y_t=y}.\label{eq:expectation}
	\end{align}
	We evaluate this term using Monte Carlo simulation. In general this is a nested Monte Carlo problem, as  one needs also to compute the default free value of the reinsurance contract $v^\phi (t, \widetilde L_t, \widetilde X_{t}+\gamma \widetilde X_{t})$, for every $0 \leq t \leq T$. For the case where $\kappa=\sigma=0$, $\widetilde L$ follows a compound Poisson process and we may use Panjer recursion. For the general case, we mostly use a regression-based approach to reduce the computational cost (see, \citet[Chapter 8.6]{glasserman2003monte}). The computation of the mean-variance minimizing hedging strategies involves computing derivatives of the functions $f^{\CVA}$ and $g$. These are computed via a Monte Carlo approach, following \citet[Chapter 7.2]{glasserman2003monte}.

	\subsection{CVA and wrong-way risk} \label{sec:numericsCVAvalue}
	
	In this section we analyse the impact of the pricing contagion and the correlation between the loss and the default intensities on the CVA by varying the parameters $\gamma$ and $\rho$. We assume that $\sigma=0.2$ and that claim sizes are Gamma(1,1) distributed.

	In Figure \ref{fig:comp_rho} we display the CVA at time $0$ for different values of $\gamma\in [0,1]$ (left panel) and for different correlation levels $\rho\in [0,1]$ (right panel). In these plots we fixed $\kappa=0.5$.	We see that $\CVA_0$ increases in both $\rho$ and $\gamma$, which is in line with the observation in Remark \ref{rem:wwrisk}.  The effect of price contagion (i.e. variation in $\gamma$) is quite pronounced and dominates the effect of dependence between intensities (i.e. variation in $\rho$), and we conclude that it is very important to incorporate price contagion into the analysis of RCCR.
	
	\begin{figure}[h]
		\centering
		\begin{tabular}{@{}c@{\hspace{.5cm}}c@{}}
			\includegraphics[width=1\textwidth]{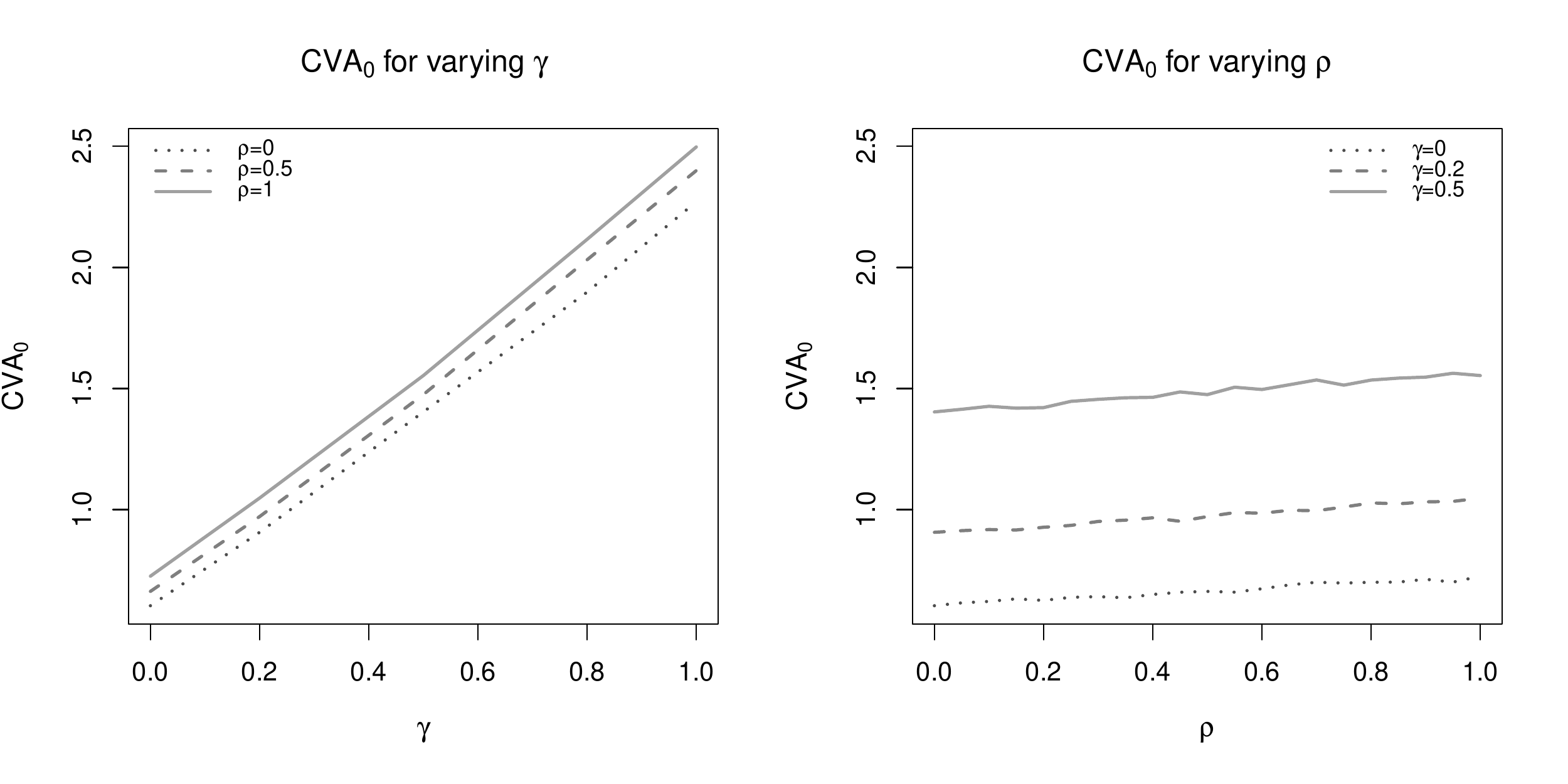}
		\end{tabular}
		\caption{ Left: $\CVA_0$ for varying contagion parameter $\gamma$. Right: $\CVA_0$ for varying correlation $\rho$.}
		\label{fig:comp_rho}
	\end{figure}

	

	\subsection{Performance of hedging strategies}\label{sec:performance}
	We now compute the hedging strategies corresponding to different parameter choices and we compare their performance to that of a static strategy. Precisely we consider the  three cases described in Table \ref{tab:parameters} below.
Case 1 and Case 2 correspond to a loss intensity that stays constant with a single jump at time $\tau_R$, where it increases by $20\%$. The parameters of the claims size distribution and the loss intensity are chosen in such a way that the expected contagion-free loss is the same ($\mathbb{E}^\Q \big [\widetilde L\big ]= 100$). However in Case 1 the insurance company experiences small but frequent losses whereas in Case 2 there are infrequent but large losses. Intuitively we therefore expect hedging to be more difficult in the second case.
	\begin{table}[h]
		\begin{center}
			\begin{tabular}{l c c c c c c c }
				& $X_0$ & $\gamma$ & $\kappa$ & $\sigma$ &  $\rho$ & $\alpha$ & $\beta$ \\
				\hline
				\textbf{Case 1:} 			& 100 & 0.2 & 0 & 0 & 0 & 1 & 1  \\
				\textbf{Case 2:} 			& 10 & 0.2 & 0 & 0 &  0 & 10 & 1 \\
				\textbf{Case 3:} 			& 100 & 0 & 1 & 0.2  &  0.2 & 1 & 1  \\
				\hline
			\end{tabular}
		\end{center}
		\vspace{2mm}
		\caption{Parameters used in the analysis of the hedging strategies. Recall that the claim sizes are Gamma$(\alpha,\beta)$ distributed.}  \label{tab:parameters}
	\end{table}

	In addition to the dynamic $\Q$-mean-variance minimizing strategies from Theorem \ref{thm:strategy} we considered two simpler strategies. First we considered a \emph{static CDS hedging strategy} where the value of the CVA at $t=0$ is invested in the CDS  and where  the position is not adjusted over time (in mathematical terms $V_0(\xi) = \CVA_0$ and  $\xi^1_t=\frac{\CVA_0}{\zeta}, 0\leq t \leq \tau^R \wedge T)$.  Moreover we considered  a strategy labelled \emph{unhedged CVA},  where the amount $\CVA_0$ is invested  in the bank account and where one does not invest in the CDS at all ($V_0(\xi) = \CVA_0$ and  $\xi^1_t\equiv 0$). In order to measure the performance of a hedging strategy we consider the value of the hedged CVA position, which is  given by
\begin{align}\label{eq:tracking_error}
e_t:=\CVA_t-\Big(\CVA_0+\int_0^t  \xi^1_s \ud S_s\Big), \quad 0 \leq t \leq T\,.
\end{align}
In the sequel we refer to the process $(e_t)_{0 \leq t \leq T}$ in \eqref{eq:tracking_error} as the {\em tracking error}. Note that a positive value of $e_T$ corresponds to a loss for the insurance company. In our experiments we assume that the hedging portfolio is re-balanced approximately every two weeks. More frequent re-balancing is not practically feasible for insurance companies as the total claim amount is hard to evaluate.

In Figure \ref{fig:comp_ex1} we use the parameter set corresponding to Case 1. The plot displays 2000 trajectories of the tracking error,	first for $\xi^1=0$ (unhedged CVA), second for the static CDS strategy $\xi^1=\CVA_0/\zeta$ and third for the dynamic $\Q$-mean-variance minimizing strategy $\xi^1=\xi^{1,*}$ from Theorem \ref{thm:strategy}.

	\begin{figure}[hbtp]
		\centering
		\begin{tabular}{@{}c@{\hspace{.5cm}}c@{}}
			\includegraphics[width=.88\textwidth]{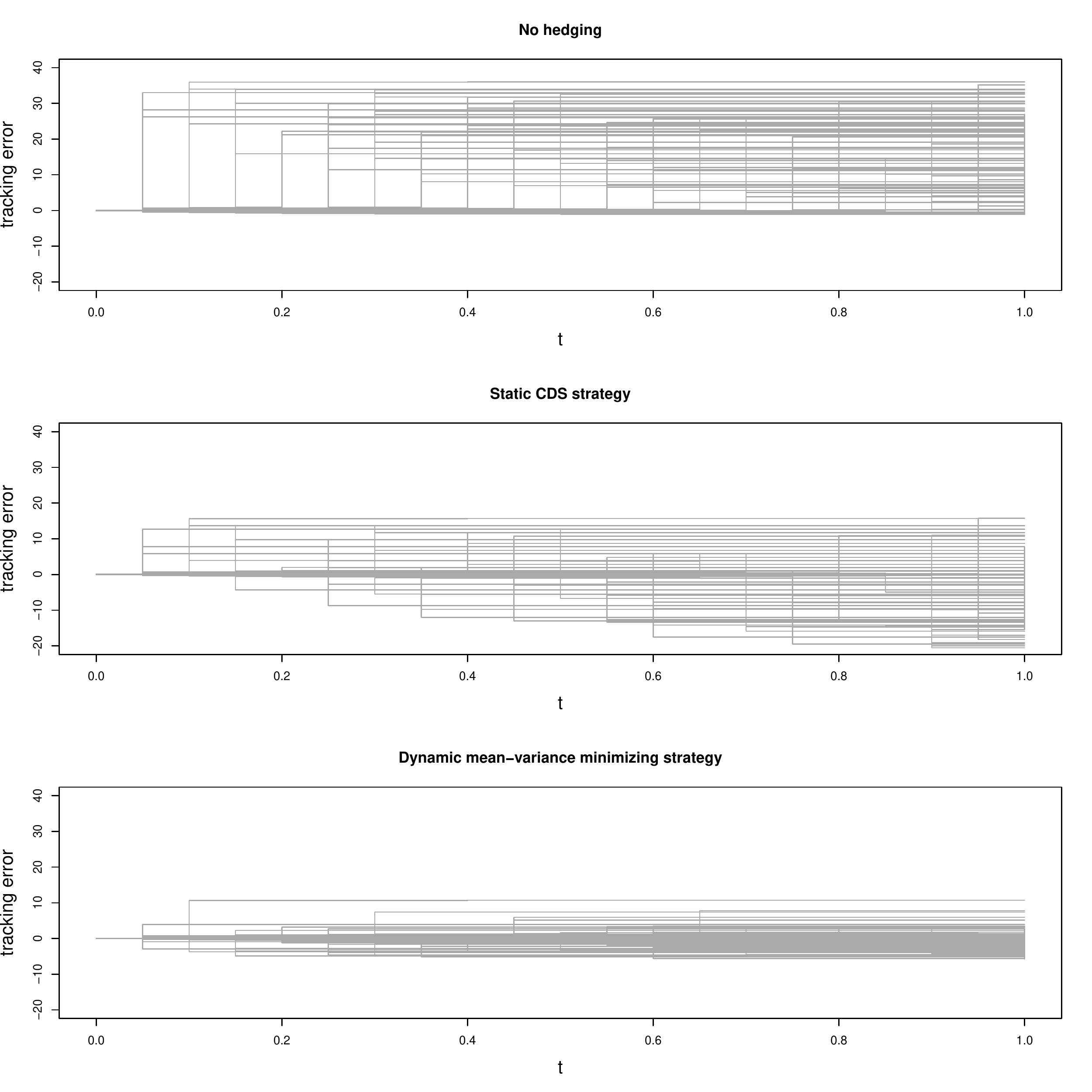}
		\end{tabular}
		\caption{Performance of various hedging strategies for the parameters in Case 1: the upper panel corresponds to no hedging, the middle panel to  static hedging and the lower panel to  dynamic mean-variance hedging.}
		\label{fig:comp_ex1}
	\end{figure}
	
From Figure \ref{fig:comp_ex1} it is evident that for all three strategies the tracking error jumps at $\tau_R$, but the form  of the jumps is  very different. In the unhedged-CVA case  the jump is always upwards and the size of the jump is equal to the replacement cost for the reinsurance contract. In this case a default of R is relatively expensive: the maximum loss that the insurance company incurs  is around EUR $40$, which is roughly three  times the initial value of the reinsurance contract. In the middle panel we give the tracking error for the static CDS hedging strategy. We observe  either a loss (under-hedging) or a profit (over-hedging). The maximum loss (and profit) is around EUR $20$ which implies that static hedging is an improvement over the unhedged CVA , but the tracking error still shows a high variability. The dynamic mean-variance minimizing strategy on the other hand significantly reduces the variability of the tracking error as it is clearly displayed in the lower panel. We conclude that this strategy out-performs the other hedging approaches by a large margin.
The difference in the performance of the hedging strategies is illustrated further in Figure \ref{fig:density_ex1} where we plot the density of the tracking error $e_T$ conditional on $\{\tau_R<T\}$. For a good hedging strategy the density of the tracking error should be concentrated around zero with a small mass in the tails. This is the case for the mean-variance minimizing strategy. The densities for the two other strategies have much larger mass in the tails. The shape of these densities is identical, but that corresponding to the static CDS strategy is shifted to the left, which results in a lower value of $\esp{e_T^2}$. The value of the $L^2$-norm of $e_T$ for all three strategies is given in Table \ref{tab:1}.
	\begin{table}[h]
		\begin{center}
			\begin{tabular}{l   c }
				Strategy& $\esp{e_T^2}$ \\ \hline
				No hedging  & 22.65 \\
				Static CDS hedging & 4.54 \\
				Dynamic mean-variance minimizing & 0.62
			\end{tabular}
		\end{center}
		\vspace{2mm}
		\caption{$L^2$-norm of the tracking error $e_T$ in Case 1.} \label{tab:1}
	\end{table}
	
	\begin{figure}[h]
		\centering
		\begin{tabular}{@{}c@{\hspace{.5cm}}c@{}}
			\includegraphics[width=.82\textwidth]{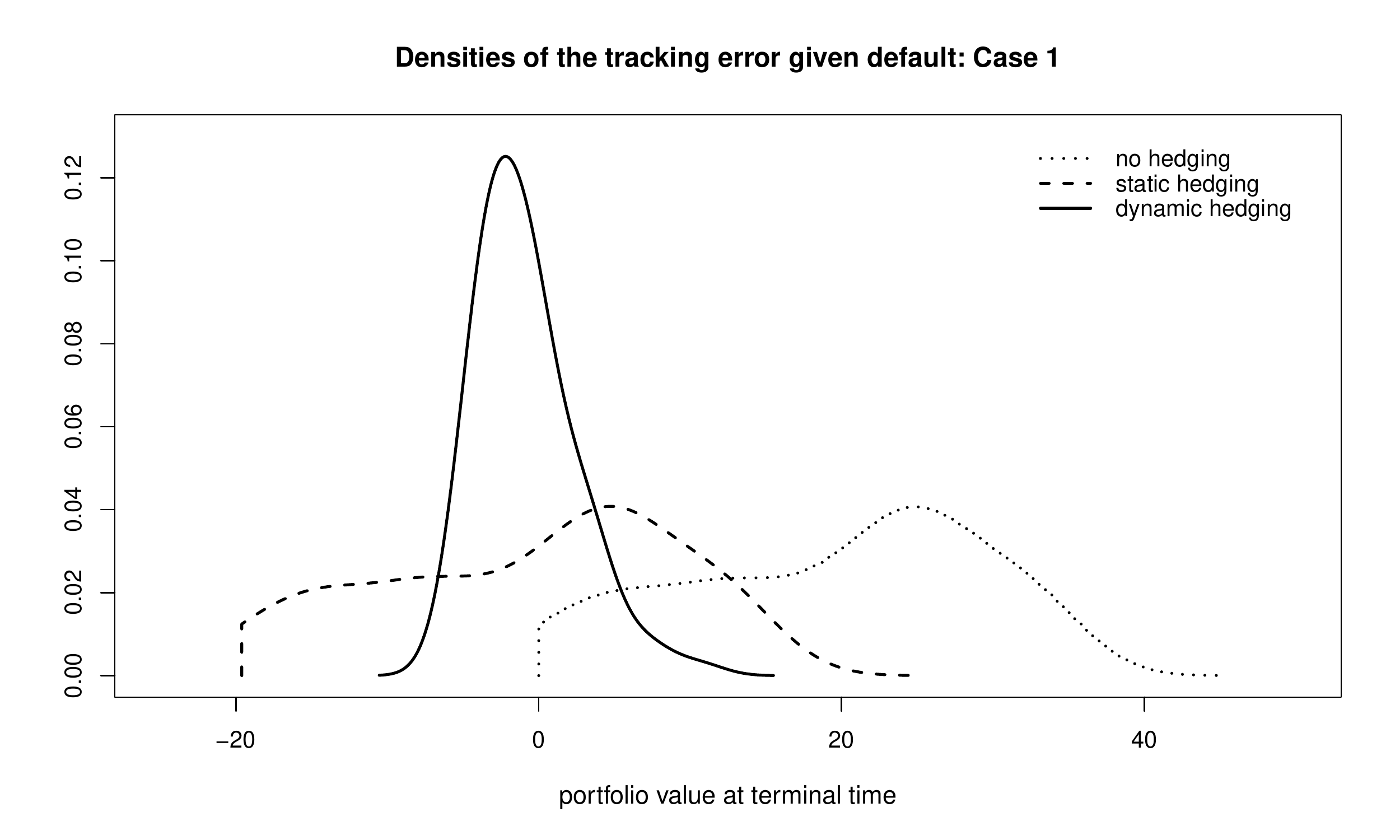}
		\end{tabular}
		\caption{Densities the tracking error $e_T$ given default in Case 1.}
		\label{fig:density_ex1}
	\end{figure}

	In order to explain the superior performance of the dynamic strategy we plot in Figure \ref{fig:comparison_scenarios} two trajectories  $\xi^{1,*}_{\cdot}(\omega)$  of the optimal strategy. The solid line corresponds to a trajectory of the claim amount process with a large loss, the dashed line to a trajectory with small loss. We compare these strategies to the static hedging strategy which is constant over time  (grey line). We see that the optimal hedge ratio is quite sensitive with respect to the evolution of the underlying loss process.
		\begin{figure}[h]
		\centering
		\begin{tabular}{@{}c@{\hspace{.5cm}}c@{}}
			\includegraphics[width=.82\textwidth]{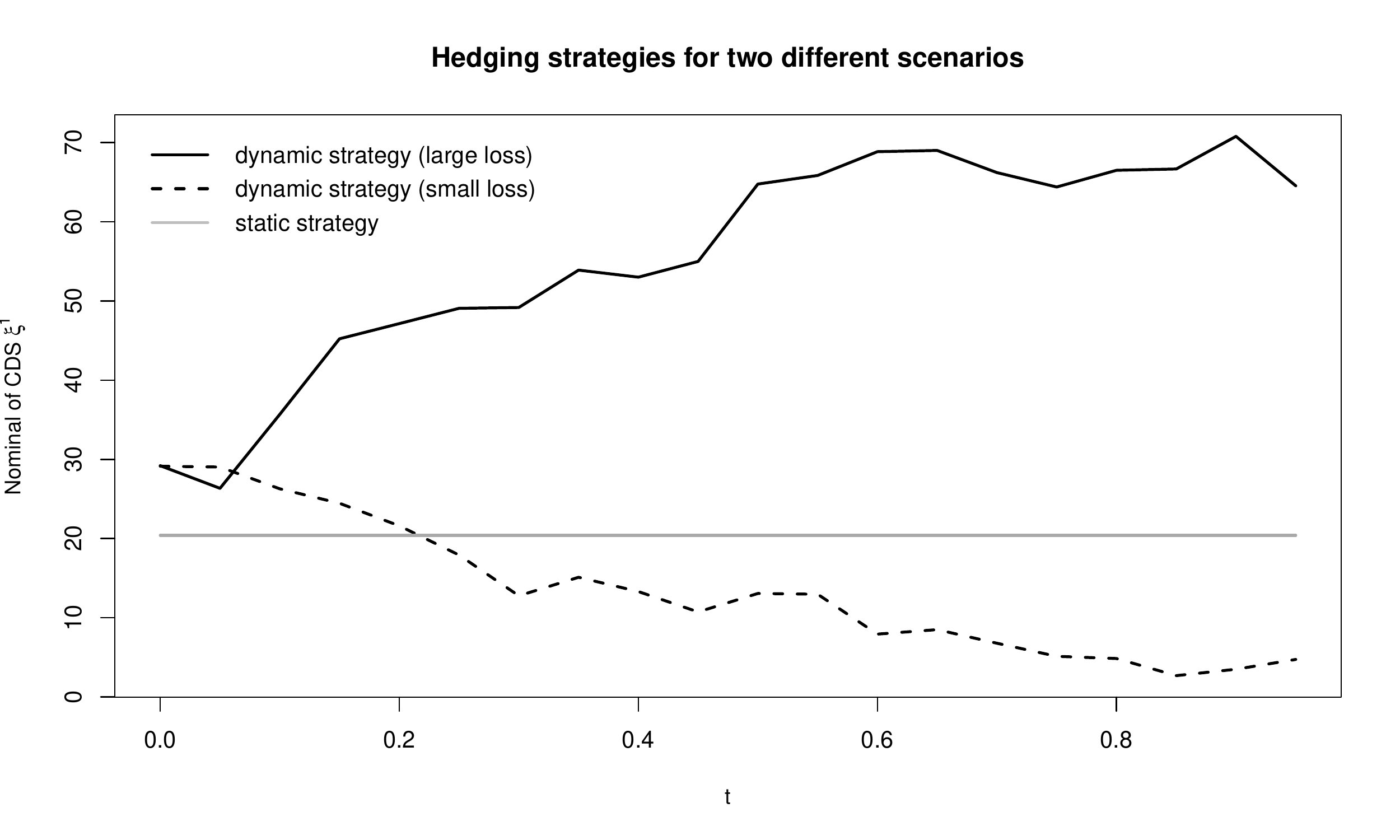}
		\end{tabular}
		\caption{Optimal strategies for two scenarios with a large loss and a low loss respectively and the constant strategy for the parameter in Case 1.}
		\label{fig:comparison_scenarios}
	\end{figure}

In Case 2 we consider the situation where claims arrive less frequently but have  on average a higher size.
In this case  hedging is more difficult, but  the mean-variance minimizing strategy still outperforms the other approaches,  as  is clearly seen  from Figure \ref{fig:density_ex2}. Moreover, for the mean-variance minimizing strategy  the $L^2$-norm of the tracking error is considerably smaller than for the other strategies, see  Table \ref{tab:2} for details. 	In Case 3 we consider the situation where the loss and the default intensities are correlated but there is no pricing contagion ($\gamma=0$), that is the loss intensity does not jump at time $\tau_R$.
	Here the wrong way risk arises from correlation only. Figure \ref{fig:density_ex3} confirms the relative performance of the strategies for this case as well. In the general model with price contagion and correlation the qualitative results on the behaviour of the tracking error are similar to the ones described so far; we omit the details.

Summarizing, our results show that dynamic CDS trading strategies  have the potential to significantly reduce reinsurance counterparty risk, both compared to a static hedging strategy  and to the case where the insurance company does not hedge at all.
	
	\begin{table}[htpb]
		\begin{center}
			\begin{tabular}{l    c }
				Strategy & $\esp{e_T^2}$ \\ \hline
				No hedging  & 39.78 \\
				Static CDS hedging & 17.82 \\
				Dynamic mean-variance minimizing & 2.17
			\end{tabular}
		\end{center}
		\vspace{2mm}
		\caption{$L^2$-norm of the tracking error in Case 2.} \label{tab:2}
	\end{table}

	\begin{figure}[h]
		\centering
		\begin{tabular}{@{}c@{\hspace{.5cm}}c@{}}
			\includegraphics[width=.82\textwidth]{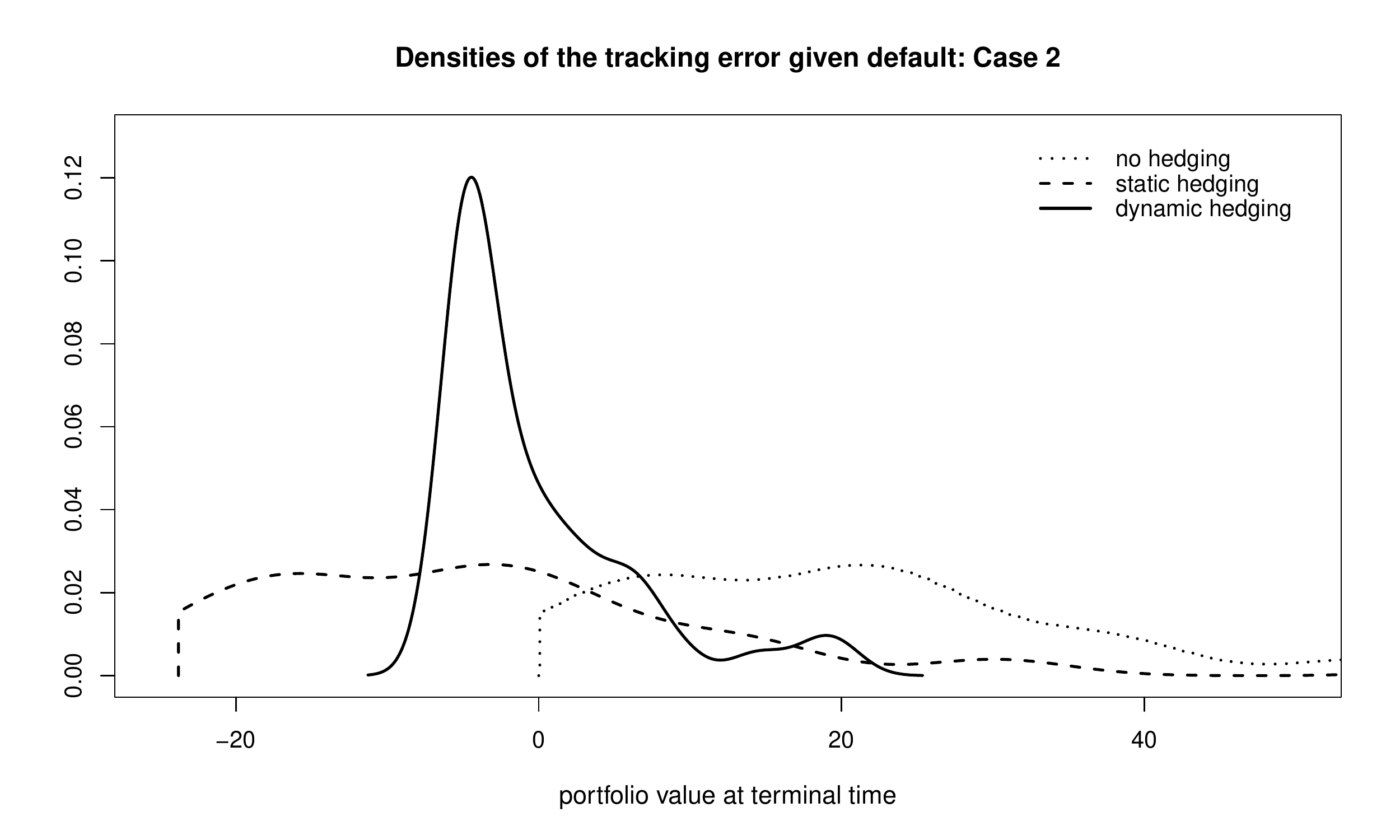}
		\end{tabular}
		\caption{Densities the tracking error at terminal time given default in Case 2.}
		\label{fig:density_ex2}
	\end{figure}


	\begin{figure}[h]
		\centering
		\begin{tabular}{@{}c@{\hspace{.5cm}}c@{}}
			\includegraphics[width=.82\textwidth]{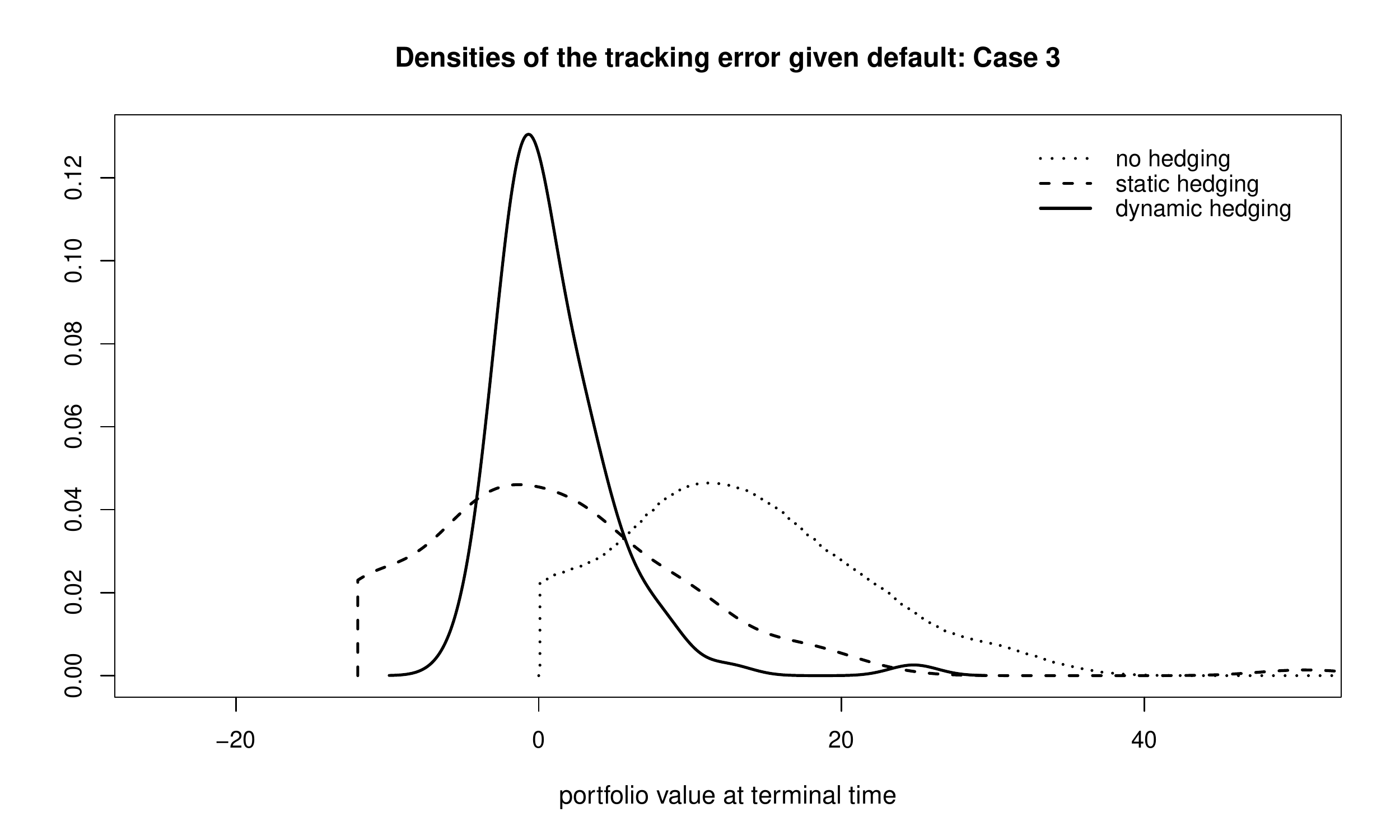}
		\end{tabular}
		\caption{Densities the tracking error at terminal time given default in Case 3.}
		\label{fig:density_ex3}
	\end{figure}
	
\clearpage

	\section*{Acknowledgements}
The authors are grateful for useful comments from Hansj{\"o}rg Albrecher.
Support by the Vienna Science and Technology Fund (WWTF) through project MA14-031 is gratefully acknowledged.
The work of K.~Colaneri was partially supported by INdAM GNAMPA through the project
 UFMBAZ-2018/000349. A part of this article was written while K.~Colaneri was affiliated with the School of Mathematics, University of Leeds, LS2 9JT, Leeds, UK. The work of the first and the second author was partially supported by INdAM-GNAMPA through projects UFMBAZ-
2017/0000327 and  UFMBAZ-2019/000436.

	\appendix
	\section{The martingales $M^{\CL}$ and $S$}\label{app:martingales}\label{Appendix1}
	In the sequel we provide detailed computations for the dynamics of the martingale $M^{\CL}$.
	We start with the martingale $M^{\CL}$. For every $0 \le t \le T$ we have that
	$$
	M^{\CL}_t = \int_0^t e^{-rs} v^{\phi}(s,\widetilde L_{s^-}, \widetilde X_{s^-}+\gamma^{X}(\widetilde X_{s^-})) \ud H_s^R + e^{-rt} \delta^R (1-H_t^R) f^{\CVA} (t,\widetilde L_t, \widetilde X_t, Y_t),
	$$
	so that
	\begin{align}
	\ud M^{\CL}_t&= e^{-rt} (v^\phi(t,\widetilde L_{t^-},\widetilde X_{t^-} + \gamma^X(\widetilde X_{t^-}) -f^{\CVA}(t, \widetilde L_{t^-}, \widetilde X_{t^-}, Y_t))\ud H^R_t \\
	& -\ re^{-rt} (1-H^R_{t-}) f^{CVA}(t,\widetilde L_{t-}, \widetilde X_t, Y_t) \ud  t + e^{-rt} (1-H^R_{t-}) \ud f^{\CVA}(t,\widetilde L_t,\widetilde X_t, Y_t)\,.
	\end{align}
	Recall that by Proposition \ref{prop:CVA}, $f^{\CVA}$ is a smooth solutions of the PIDE \eqref{eq:PIDE2}, therefore it has the necessary regularity to apply the It\^{o} formula. This gives
	\begin{align}
	\ud f^{\CVA} &(t,\widetilde L_t, \widetilde X_t, Y_t) =
	\left( \frac{\partial f^{\CVA}}{\partial x }(t,\widetilde L_{t-}, \widetilde X_{t}, Y_t) \sigma^X(\widetilde X_{t})+\frac{\partial f^{\CVA}}{\partial y }(t, \widetilde L_{t-}, \widetilde X_{t}, Y_t) \sigma^Y(Y_t) \rho\right) \ud W^1_t\\
	&+\frac{\partial f^{\CVA}}{\partial y }(t, \widetilde L_{t-}, \widetilde X_t, Y_t) \sigma^Y(Y_t) \sqrt{1-\rho^2} \ud W^2_t \\
	&+ \int_{\R^+} \left(f^{\CVA}(t, \widetilde L_{t-}+z, \widetilde X_{t}, Y_{t})-f^{\CVA}(t, \widetilde L_{t-}, \widetilde X_{t}, Y_{t}) \right) m^L(\ud s, \ud z)\\
	&+  \left(\frac{\partial f^{\CVA}}{\partial t}(t, \widetilde L_{t-}, \widetilde X_{t}, Y_t) + b^X(\widetilde X_t)\frac{\partial f^{\CVA}}{\partial x}(t, \widetilde L_{t-}, \widetilde X_{t}, Y_t) \right.\\
	& \quad  + b^{Y} (Y_t)\frac{\partial f^{\CVA}}{\partial y}(t,\widetilde L_{t-},\widetilde X_{t}, Y_t)  + \frac{1}{2}(\sigma^X(\widetilde X_t))^2 \frac{\partial^2 f^{\CVA}}{\partial x^2}(t,\widetilde  L_{t-},\widetilde X_{t}, Y_t) \\
	& \quad \left. + \frac{1}{2}(\sigma^Y(Y_t))^2 \frac{\partial^2 f^{\CVA}}{\partial y^2}(t,\widetilde L_{t-}, \widetilde X_{t}, Y_t)  + \rho\sigma^X(\widetilde X_t)\sigma^Y(Y_t) \frac{\partial^2 f^{\CVA}}{\partial x \partial y}(t,\widetilde L_{t-}, \widetilde X_{t}, Y_t) \right) \ud t.
	\end{align}
	Now using the fact that $f^{\CVA}$ solves equation \eqref{eq:PIDE2} we get that $M^{\CL}$ satisfies equation \eqref{eq:MCL}. Similar computations can be performed for the martingale $S$, we omit the details.

	\bibliographystyle{plainnat}
	\bibliography{reinsurance_bib}

\end{document}